\documentclass[12pt]{article}
\usepackage{amsmath,amsthm,amssymb,mathrsfs,url,pict2e,hyperref,tikz}
\usepackage{xcolor}

\title{Multi-Time Version of the Landau-Peierls Formulation of Quantum Electrodynamics}
\author{
Matthias Lienert\footnote{Fachbereich Mathematik, Eberhard-Karls-Universit\"at T\"ubingen, Auf der Morgenstelle 10, 72076 T\"ubingen, Germany.}~\footnote{E-mail: lienertmat@gmail.com}\ \ and
Roderich Tumulka$^*$\footnote{E-mail: roderich.tumulka@uni-tuebingen.de}
}
\date{October 5, 2024}

\renewcommand{\Im}{\mathrm{Im}}
\newcommand{\NNN}{\mathbb{N}}
\newcommand{\RRR}{\mathbb{R}}
\newcommand{\CCC}{\mathbb{C}}
\newcommand{\SSS}{\mathbb{S}}
\newcommand{\MMM}{\mathbb{M}}
\newcommand{\free}{\mathrm{free}}
\newcommand{\ext}{\mathrm{ext}}

\newcommand{\Hilbert}{\mathscr{H}}

\newcommand{\sA}{\mathscr{A}}
\newcommand{\sD}{\mathscr{D}}
\newcommand{\sF}{\mathscr{F}}
\newcommand{\sG}{\mathscr{G}}
\newcommand{\sJ}{\mathscr{J}}

\newcommand{\sS}{\mathscr{S}}
\newcommand{\sT}{\mathscr{T}}
\newcommand{\cI}{\mathcal{I}}

\newcommand{\scp}[2]{\langle #1|#2 \rangle}

\newcommand{\vk}{\boldsymbol{k}}

\newcommand{\vu}{\boldsymbol{u}}
\newcommand{\vx}{\boldsymbol{x}}
\newcommand{\vy}{\boldsymbol{y}}

\newcommand{\vA}{\boldsymbol{A}}
\newcommand{\vB}{\boldsymbol{B}}
\newcommand{\vE}{\boldsymbol{E}}
\newcommand{\vJ}{\boldsymbol{J}}

\newcommand{\vX}{\boldsymbol{X}}

\newcommand{\valpha}{\boldsymbol{\alpha}}
\newcommand{\vomega}{\boldsymbol{\omega}}
\newcommand{\vzero}{\boldsymbol{0}}
\newcommand{\be}{\begin{equation}}
\newcommand{\ee}{\end{equation}}
\newcommand{\Gr}{\mathrm{Gr}}

\newtheorem{prop}{Proposition}
\newtheorem{conj}{Conjecture}
\theoremstyle{definition}\newtheorem{defn}{Definition}
\newcounter{remarks}
\newcommand{\un}{\unitlength}
\newcommand{\spacelike}{\begin{minipage}[b]{7\un}
\begin{picture}(6,6)
\put(0,5){\line(1,-1){6}}
\put(0,-1){\line(1,1){6}}
\end{picture}
\end{minipage}}

\begin{document}
\maketitle
\begin{abstract}
Landau and Peierls wrote down the Hamiltonian of a simplified version of quantum electrodynamics in the particle-position representation. We present a multi-time version of their Schr\"odinger equation, which bears several advantages over their original equation: the time evolution equations are simpler and more natural; they are more transparent with respect to choice of gauge; and, perhaps most importantly, they are manifestly Lorentz covariant. We discuss properties of the multi-time equations. Along the way, we also discuss the Lorentz covariant 3d Dirac delta distribution for spacelike surfaces and the inner product of photon wave functions on spacelike surfaces in an arbitrary gauge.

\bigskip

\end{abstract}
\tableofcontents

\section{Introduction}

In 1930, Landau and Peierls \cite{LP30} arrived at remarkably simple and natural equations for quantum electrodynamics (QED) by using the particle-position representation. We argue here that their equations become even simpler and more natural when written as multi-time equations.
We present a system of multi-time equations that reduces to the Schr\"odinger equation (and a constraint equation) written by Landau and Peierls when all time coordinates are set equal. (See \cite{LPT17,LPT20} for introductions to multi-time wave functions.) Our equations are manifestly Lorentz invariant (which is not possible for single-time equations) and allow for a new, clearer perspective on choices of gauge. 
In addition, while our equations are, like those of Landau and Peierls, ultraviolet divergent and thus mathematically ill defined, we outline a way of interpreting them mathematically by means of interior-boundary conditions. (See \cite{TT15a,TT15b} for introductions to interior-boundary conditions.)

The model of Landau and Peierls describes the emission and absorption of photons by electrons while leaving out pair creation and annihilation, in fact leaving out positrons altogether. An electron wave function is taken to be governed by the Dirac equation, a photon wave function by the complexified Maxwell equation \cite{BB}. A remarkable (and in our view convincing) trait of this model is that the photon part of the wave function, when expressed as a vector potential, plays the role of the gauge connection that is used in the version of the Dirac equation governing the electron part. No steps are taken here to exclude wave functions of negative energy; the issues of negative energies and the Dirac sea are left aside. As in Fock space, the full wave function $\Psi$ consists of sectors $\Psi^{(m,n)}$ with arbitrary $m,n\in\NNN\cup\{0\}$ for $m$ electrons and $n$ photons. Different sectors are coupled through particle creation and annihilation terms in the time evolution equations. The model is thus also a variant of the Pauli-Fierz model of QED \cite{PF38,GLL01,Hir02}, in fact one that is intended to be fully relativistic, although Landau and Peierls themselves could neither prove nor even clearly formulate what it means for the model to be relativistically invariant. This is now possible in our multi-time formulation, in fact straightforwardly. Our equations can also be regarded as a further development of the models proposed in \cite{dfp:1932,DV85,pt:2013c,Kie20}. 

As a by-product of our considerations, we also describe a Lorentz-covariant version $\delta^3_\mu$ of the 3-dimensional Dirac delta distribution that can be integrated over any smooth spacelike 3-surface through the origin of $\RRR^4$ (see Section~\ref{sec:delta}). We also provide a discussion of the inner product between photon wave functions in any gauge on any Cauchy surface (see Appendix~\ref{app:GuptaBleuler}), which seems to have been missing in the literature.

Just as a classical electromagnetic field can be described by either the field tensor $F_{\mu\nu}$ or the vector potential $A_{\mu}$, the wave function in the multi-time formulation can be written in a representation analogous to $F_{\mu\nu}$ or one analogous to $A_{\mu}$ (see Section~\ref{sec:AF}); in fact, we can choose between these representations for each photon separately. While Landau and Peierls fixed a gauge, we keep the choice of gauge open and include a discussion of gauge invariance in Section~\ref{sec:gauge}. In particular, we describe how the wave functions $\Psi$ in the $A_\mu$ representation transform under a change of the gauge condition.

The equations developed here will be particularly attractive for those who view the wave function $\Psi$ as part of reality, like a field (except that it is a function of several space-time points rather than one) rather than like an abstract vector in an abstract Hilbert space. The picture we develop has concrete answers to what kind of function $\Psi$ is, in which spaces its values lie, and how it transforms under Lorentz and gauge transformations. This trait also allows for a direct generalization to curved space-time.

Historically, the particle-position representation was largely neglected in quantum field theory (QFT) (but see, e.g., \cite{Nel64}), and Landau and Peierls themselves did not follow up on it further. That was perhaps due to belief that ultimately the particle-position picture does not work, and that may in part be owed to the problem, still open today, of formulating the Born rule for photons, i.e., of computing the probability density of a photon particle position from an arbitrary quantum state \cite[Sec.~7.3.9]{Tum22}. A convincing candidate for the Born rule is known for photon wave functions that are plane waves, or at least local plane waves (i.e., functions such that every point has a neighborhood in which the function is approximately a plane wave).
While that case may cover most experiments, not every Maxwell field is a local plane wave, and the question of the general Born rule for photons remains open. 
Nevertheless, we think that the particle-position picture was abandoned too hastily. Although we do not have a Born rule for photons, 
we tend to believe that one exists because we can experimentally detect individual photons in places of our choice---so photons appear to have a probability distribution in position space. But even if somehow no Born rule for photons existed, it would still be possible to represent the quantum state as a function of position variables, and the particle-position picture would still seem worthy of exploration. For the purpose of this paper, which is to formulate and study equations for the time evolution of the quantum state, it is not necessary to specify the probability distribution of the photon positions. 

Our equations can directly and easily be generalized to curved space-time, but for simplicity we will only formulate them for Minkowski space-time, henceforth denoted by $\MMM^4$ (where 4 is not an exponent but an upper index indicating the dimension). In the future, it would be of interest to extend the approach presented here to Yang-Mills theory or in general to non-Abelian gauge theories, as well as to include electron-positron pair creation and annihilation. It would also be of interest to obtain a complete model by adding an ontology; see \cite{Tum24} for attempts in this direction. See also \cite{Lie18} for a possible relativistic time evolution in which fermions interact directly without mediation through bosons, and \cite{Fin16} for another approach to the foundations of relativistic quantum theory.

The paper is organized as follows. In Section~\ref{sec:multi}, we formulate the set of multi-time equations. In Section~\ref{sec:AF}, we describe the representations of the wave function analogous to $A_\mu$ and $F_{\mu\nu}$. In Section~\ref{sec:gauge}, we discuss gauge invariance. In Section~\ref{sec:field}, we elucidate a derivation of the multi-time equations from the equations governing the field operators. In Section~\ref{sec:scp}, we point out that the multi-time evolution preserves the scalar product of wave functions (a generalized form of unitarity). In Section~\ref{sec:LP}, we show that for equal times, our equations reduce to those of Landau and Peierls (up to minor corrections). In Section~\ref{sec:consistency}, we discuss the question of consistency and of preservation of constraints under the time evolution. In the Appendix, we collect technical arguments.

\section{Multi-Time Equations}
\label{sec:multi}

\subsection{Wave Function}

A multi-time wave function is a function of a space-time configuration, i.e., of several space-time points. A space-time point will be denoted interchangeably by $x=x^\mu=(x^0,\vx)=(x^0,x^1,x^2,x^3)$. A space-time point that is an electron variable will be denoted $x$, one for a photon $y$. For a space-time configuration $(x_1,\ldots,x_m)$, we will write $x^{4m}$ (with the superscript reminding us that it has $4m$ components), and likewise $y^{4n}$ for $(y_1,\ldots,y_n)$. The notation $x_1\spacelike x_2$ means that $x_1$ and $x_2$ are spacelike separated. We say that $(x^{4m},y^{4n})$ is a \emph{spacelike} configuration if any two of $x_1,\ldots,x_m,y_1,\ldots,y_n$ are either spacelike separated or equal. Let $\sS_{xy}$ be the set of all spacelike configurations with any $m,n\in\NNN\cup\{0\}$ and $\sS_{xy}^{(m,n)}$ the set of spacelike configurations with given $m$ and $n$, called the $(m,n)$-sector of $\sS_{xy}$. The wave function $\Psi$ will be a function on $\sS_{xy}$, and its restriction $\Psi^{(m,n)}=\Psi\big|_{\sS_{xy}^{(m,n)}}$ will be called the $(m,n)$-sector of $\Psi$. We write $\sS_x=\cup_{m=0}^\infty \sS_x^{(m)}$ for the set of all spacelike configurations of $x$-particles alone.

We take the wave function of an electron to be a Dirac wave function and the wave function of a photon a complexified Maxwell field, represented by a complex 4-vector $A_\mu$. Thus, the value space of $\Psi^{(m,n)}$ is the complex tensor product of $m$ copies of the (4-dimensional) Dirac spin space and $n$ copies of complexified Minkowski space (or the appropriate tangent spaces of the space-time manifold); that is, $\Psi^{(m,n)}$ has, at each configuration $(x^{4m},y^{4n})$, $4^{m+n}$ complex components labeled by $m+n$ indices as in
\be\label{Psiformat}
\Psi^{(m,n)}_{s_1...s_m,\mu_1...\mu_n}(x_1...x_m,y_1...y_n)\,,
\ee
with $s_j$ running through $\{1,2,3,4\}$ and $\mu_k$ through $\{0,1,2,3\}$. It is assumed of the initial data and will turn out for all other configurations that each sector $\Psi^{(m,n)}$ is symmetric against permutation of the $y$'s with simultaneous permutation of the associated $\mu$ indices, and anti-symmetric against permutation of the $x$'s together with the $s$'s. 

Before we formulate the multi-time equations, we need a bit of preparation concerning 3-dimensional Dirac delta distributions.

\subsection{3d Dirac delta Distribution}
\label{sec:delta}

The ordinary 3-dimensional Dirac delta distribution $\delta^3(\vx)$ can be integrated over 3-space against smooth functions $f(\vx)$ on 3-space,
\be
\int_{\RRR^3}d^3\vx \: \delta^3(\vx) \, f(\vx) = f(\vzero)\,.
\ee
We will define a Lorentz-invariant distribution $D$ that can be integrated over any smooth spacelike 3-surface $\Sigma\subset \MMM^4$ passing through the origin of Minkowski space-time $\MMM^4$ and acts like $\delta^3$ on $\Sigma$. To this end, let
\be
\sS_0^4 = \{x\in\MMM^4: x=0 \text{ or }x\spacelike 0\}
\ee
be the relativistic present of the origin in Minkowski space-time together with the origin itself; $D$ will be defined on $\sS_0^4$. 

As a preparation we note, leaving the distributional character of $D$ aside for a moment and pretending it was an ordinary function, that the kind of object that could be integrated over any 3-surface is a 3-form, i.e., an anti-symmetric rank-3 tensor field $D_{\lambda\mu\nu}$. Using the 4-form (volume form) $\varepsilon_{\lambda\mu\nu\rho}$ provided by the space-time metric (and the orientation of space-time), any 3-form can be translated into a vector field and vice versa,
\be
D^\rho = \tfrac{1}{3!} \, D_{\lambda\mu\nu} \, \varepsilon^{\lambda\mu\nu\rho}\,,~~~
D_{\lambda\mu\nu} = D^\rho \, \varepsilon_{\lambda\mu\nu\rho}\,.
\ee 
That is, a 3-dimensional distribution (i.e., one which can be integrated over 3-surfaces) on the 4-dimensional set $\sS_0^4$ can be represented as a \emph{vector-valued} distribution $D^\rho$. The integral of the 3-form $D_{\lambda\mu\nu}$ times a scalar function $f$ over a spacelike surface $\Sigma$ can be expressed in terms of the vector field $D^\rho$ as
\be
\int_\Sigma V(d^3x) \, n_\rho(x) \, D^\rho(x) \, f(x)\,,
\ee
where $V(d^3x)$ means the Riemannian volume (defined by the 3-metric on $\Sigma$) of the volume element $d^3x$ and $n_\rho(x)$ the future unit normal vector to $\Sigma$ at $x$.

The distribution $D$ is characterized by the property that for every (smooth) spacelike Cauchy surface $\Sigma$ passing through the origin and every smooth function $f:\Sigma\to\CCC$, $D$ applied to $f$ on $\Sigma$ ($\int_\Sigma D \, f$ for short) yields
\be\label{characterizeD}
\int_\Sigma \! D \, f=\int_\Sigma V(d^3x) \, n_\rho(x) \, D^\rho(x) \, f(x) = f(0) \,.
\ee
This prescription is clearly Lorentz invariant. In Appendix~\ref{app:delta}, we show that such a distribution exists. Henceforth, we write $\delta^3_\mu$ for $D_\mu$.

\subsection{Time Evolution Equations}

We can now formulate the fundamental equations governing the wave function. We set $c=1$, $\hbar=1$, $\varepsilon_0=1$, and $\mu_0=1$. Let $e_x$ and $m_x$ denote the (bare) charge and mass of the electron, respectively. The time evolution is defined by three laws, given by \eqref{LPx}, \eqref{LPy}, and \eqref{LPg} below, the central equations of this paper. In the notation we choose, we will omit some indices for better readability, but we will make many variables and indices explicitly visible although this might make the formulas appear more complex.

The first law provides the time derivative for each electron $x_j$ and is meant to apply to each $j$; it is the Dirac equation with an additional term:
\be\label{LPx}
(i\gamma^\mu_j\partial_{x_j,\mu}-m_x)\Psi^{(m,n)}(x_1...x_m,y_1...y_n)
= e_x\sqrt{n+1} \:\gamma^\rho_j\: \Psi^{(m,n+1)}_{\mu_{n+1}=\rho}(x_1...x_m,y_1...y_n,x_j)\,.
\ee
Here, $\gamma_j^\mu$ means $\gamma^\mu$ acting on the index $s_j$.

The second law involves derivatives in $y_k$; it is the Maxwell equation with a source term:
\begin{multline}\label{LPy}
2\partial_{y_k}^\mu\partial^{~}_{y_k,[\mu}\Psi^{(m,n)}_{\mu_k=\nu]}(x_1...x_m,y_1...y_n) =\\
\frac{e_x}{\sqrt{n}}\sum_{j=1}^m \delta^3_{\mu}(y_k-x_j) \: \gamma_j^\mu \gamma_{j\nu} \:\Psi^{(m,n-1)}_{\widehat{\mu_k}}(x_1...x_m,y_1...y_{k-1},y_{k+1}...y_n)\,.
\end{multline}
Here, $[\mu\nu]$ means anti-symmetrization in the index pair as in $S_{[\mu\nu]}=\tfrac12(S_{\mu\nu}-S_{\nu\mu})$, and $\widehat{\mu_k}$ means that the index $\mu_k$ is omitted.

Before we turn to the third law, we would like to make explicit the parallels of \eqref{LPx} and \eqref{LPy} with the corresponding 1-particle equations. The 1-particle Dirac equation in an external electromagnetic field with vector potential $A_\mu(x)$ reads
\be\label{Dirac}
(i\gamma^\mu\partial_{\mu}-m_x)\psi(x)
= e_x \,\gamma^\rho A_\rho(x) \: \psi(x)\,.
\ee
Apart from the factor $\sqrt{n+1}$, which arises from our convention about normalization of symmetric functions (and would be absent if we used unordered configurations \cite{GTTZ:2014}), \eqref{LPx} is \eqref{Dirac} applied to the $x_j$ variable in $\Psi$, with $A_\rho \Psi$ replaced by $\Psi^{(m,n+1)}_{\mu_{n+1}=\rho}$. If $\Psi^{(m,n+1)}$ factorized according to
\be
\Psi^{(m,n+1)}_{\mu_{n+1}}(x_1...x_m,y_1...y_{n+1}) = A_{\mu_{n+1}}(y_{n+1}) \: \Psi^{(m,n)}(x_1...x_m,y_1...y_n)\,,
\ee 
then \eqref{LPx} would reduce exactly to \eqref{Dirac} applied to $x_j$ (up to the $\sqrt{n+1}$). So, in \eqref{LPx} the role of the vector potential (or, more generally speaking, of the connection coefficient in the covariant derivative of a vector bundle) is played by the wave function of the next photon.

Let us turn to \eqref{LPy}. The Schr\"odinger equation for a single photon would be the complex Maxwell equation, which reads
\be\label{Maxwell}
2\partial^\mu \partial_{[\mu} A_{\nu]}(y) = J_\nu(y)
\ee
with source term $J_\nu(y)$. This equation is more often written in the form $\partial^\mu F_{\mu\nu} = J_\nu$ with $F_{\mu\nu} = 2 \partial_{[\mu}A_{\nu]}$. Apart from the factor $1/\sqrt{n}$ owed to our normalization convention, Eq.~\eqref{LPy} has the same form,
with $y_k$ playing the role of $y$ and $\Psi_{\mu_k}$ that of $A_\mu$, and the source term given by
\be\label{source}
J_\nu(y) = e_x \sum_{j=1}^m \delta^3_\mu(y-x_j) \:\gamma_j^\mu \, \gamma_{j,\nu} \: \Psi^{(m,n-1)}_{\widehat{\mu_k}}(\widehat{y_k})\,.
\ee
It is reasonable that the source term should be concentrated on the locations $x_1,\ldots,x_m$ of the electrons, and the matrix $\gamma_\nu$ links, roughly speaking, the electron wave function to the space-time direction of the associated current in the formula $j_\nu = \overline{\psi} \gamma_\nu \psi$ for the current of a 1-particle Dirac wave function $\psi$. 

These observations add to the picture that \eqref{LPx} and \eqref{LPy} are natural equations. After all, they are the Dirac and Maxwell equations applied to particular variables of $\Psi$, with natural expressions inserted for the vector potential and the source term, given that the photon aspect of $\Psi$ is analogous to $A_\mu$ and the electron aspect of $\Psi$ should be the source for the Maxwell equation. The appearance of $\Psi^{(m,n+1)}$ and $\Psi^{(m,n-1)}$ in the time derivative of $\Psi^{(m,n)}$ is common in models in which $x$-particles can absorb and emit $y$-particles according to the reaction $x\leftrightarrows x+y$.

\subsection{Gauge Condition}

We now turn to the third fundamental law of the theory, the gauge condition. As we will discuss in detail in Section~\ref{sec:gauge}, the structure of the theory is such that any gauge condition $\sA$ can be chosen, and some gauge condition has to be chosen in order to write down a wave function $\Psi$. This choice is arbitrary and does not reflect a fact in nature; it is a tool of the theoretician, like a choice of basis or of coordinates. Correspondingly, a solution of the theory is really an equivalence class of pairs $(\Psi,\sA)$, where the equivalence relation is gauge equivalence and will be defined in Section~\ref{sec:gauge}. (Earlier reports on the ideas of this paper \cite{Tum21,Tum21b,Tum24} did not have this aspect fully developed. The need for a gauge condition is discussed in Remark~\ref{rem:indispensable} in Section~\ref{sec:rem2}.)

Before we give the general characterization of ``gauge condition'' for our purposes, we provide an example of such a condition (and thus of the third law): 
\be\label{LPgex}
\sum_{\mu_k=1}^3 \partial_{y_k}^{\mu_k} \Psi^{(m,n)}_{\mu_k}(x_1...x_m,y_1...y_n)=0
\ee
(note that $\mu_k=0$ is omitted from the sum).
Again, most indices are not made explicit. The condition \eqref{LPgex} is analogous to the well-known Coulomb gauge condition
\be\label{Coulomb}
\sum_{\mu=1}^3 \partial^{\mu} A_\mu=0 \,.
\ee
Like \eqref{Coulomb}, \eqref{LPgex} is not Lorentz invariant, and our gauge conditions are not required to be Lorentz invariant. The Lorentz invariance of the theory is therefore connected to its gauge invariance, i.e., the fact that $\Psi$ can be transformed to fit any other gauge condition (such as a Lorentz transform of \eqref{LPgex}).

We now give the general definition of a gauge condition. Let us begin with the familiar fact that the Maxwell field can be represented through either the vector potential $A_\mu(x)$ or the field tensor $F_{\mu\nu}(x)=\partial_\mu A_\nu(x) - \partial_\nu A_\mu(x)$. Let
\be\label{ddef}
d_{\mu\nu}^{\rho} = \delta_\nu^\rho \partial_\mu - \delta_\mu^\rho \partial_\nu
\ee
be the differential operator (``exterior derivative $d$'') that computes $F_{\mu\nu}$ from $A_\mu$,
\be\label{FdA}
F_{\mu\nu} = d_{\mu\nu}^{\rho}\, A_\rho~~~\text{or}~~F=dA\,.
\ee
For our purposes, a \emph{gauge condition} means a set $\sA$ of complex vector fields $A_\mu$ that contains exactly one field $A_\mu$ for every complex $F_{\mu\nu}$. That is, let $\sF$ be a suitable function space (representing ``all'' $F_{\mu\nu}$'s) and demand that the exterior derivative $d$ is bijective as a mapping $\sA\to\sF$. We say that $\sA$ is a \emph{linear} gauge condition if and only if $\sA$ is a (complex) linear subspace of the space of all complex vector fields on space-time $\MMM^4$. Finally, for any (open) set $S\subset \MMM^4$, we write
\be
\sA_S=\bigl\{A|_S:A\in\sA\bigr\}
\ee
for the space of $\sA$-vector potentials in $S$; here, $A|_S$ denotes the restriction of $A$ to $S$ (i.e., the function $A$ considered only at points in $S$).

We are now ready to state the third fundamental law of the theory in its general form. Of a solution $(\Psi,\sA)$ of the theory we require that $\sA$ is a linear gauge condition and that $\Psi$ locally obeys the gauge condition $\sA$ in each photon variable $y_k$. In formulas, the latter requirement means that for every configuration $(x_1...x_m,y_1...y_n)\in \sS_{xy}$, if $S_k$ is an open neighborhood of $y_k$ for every $k$ and $S:=\{(x_1...x_m)\}\times S_1 \times \cdots \times S_n\subset \sS_{xy}$, then
\be\label{LPg}
\Psi^{(m,n)}\Big|_{S} \in (\CCC^4)^{\otimes m} \otimes \sA_{S_1} \otimes \cdots \otimes \sA_{S_n} \,,
\ee
where $\CCC^4$ represents the Dirac spin space. Condition \eqref{LPg} is the third law, so we have completed the statement of the fundamental laws of this model.

\bigskip

Let us say more explicitly how \eqref{LPgex} is an example of \eqref{LPg}. Let us express $F_{\mu\nu}$ through $\vE$ and $\vB$, that is,
\be
E^i=F_{0i}~~~\text{and}~~~B^i= -\sum_{j,k=1}^3 \varepsilon^{ijk} F_{jk}  
\ee
for $i=1,2,3$ with $\varepsilon^{ijk}$ the Levi-Civita symbol, i.e., $\varepsilon^{ijk}=\mathrm{sgn}(\sigma)$ if $ijk$ is a permutation $\sigma$ of $123$ and $=0$ otherwise. Then, as is well known, the Maxwell equation \eqref{Maxwell} reads
\begin{subequations}\label{MaxwellEB}
\begin{align}
\partial_0 \vE &= \boldsymbol{\nabla} \times  \vB -\vJ\\
\boldsymbol{\nabla} \cdot \vE&= J_0 \label{20b}\\
\partial_0 \vB &= -\boldsymbol{\nabla} \times  \vE\\
\boldsymbol{\nabla} \cdot \vB&=0 \label{20d}
\end{align}
\end{subequations}
with $\boldsymbol{\nabla}=(\partial_1,\partial_2,\partial_3)$ and $J_\nu=(J_0,-\vJ)$. The relation $F=dA$ takes the form
\be\label{EBA}
\vE=-\boldsymbol{\nabla} A_0 - \partial_0 \vA\,,~~~\vB=\boldsymbol{\nabla} \times \vA
\ee
with $A_\mu=(A_0,-\vA)$. Let $\sA$ consist of the $A_\mu$'s satisfying the Coulomb condition \eqref{Coulomb}, or $\boldsymbol{\nabla} \cdot\vA=0$. Let us take for granted that, for any fixed time $x^0$, $\vE,\vB,A_0$, and $\vA$ are square-integrable functions of $\vx$; then \eqref{EBA} can be inverted to yield
\be\label{d-1ex}
A_0 = (-\Delta)^{-1} \boldsymbol{\nabla} \cdot \vE\,,~~~\vA = (-\Delta)^{-1}\boldsymbol{\nabla} \times \vB
\ee
at every time $x^0$ with $\Delta=\partial_1^2+\partial_2^2+\partial_3^2$ the Laplacian. Thus, $d$ is bijective on $\sA$, $d^{-1}$ is given by \eqref{d-1ex}, and \eqref{LPg} amounts to \eqref{LPgex}.

\subsection{Remarks}
\label{sec:rem1}
 
\begin{enumerate}
\setcounter{enumi}{\theremarks}
\item {\it Superselection of electron number.} Since sectors with different numbers $m$ of electrons are not coupled, we can also regard this number as fixed.

\item {\it Linearity.} The laws \eqref{LPx}, \eqref{LPy}, \eqref{LPg} are complex-linear, so any linear combination of solutions (expressed in the same gauge) is another solution.

\item {\it Comparison to an external field.} A system of $m$ electrons in an external electromagnetic field with vector potential $A_\mu$ is governed (in the multi-time framework) by a wave function $\psi_{s_1...s_m}(x_1...x_m)$ obeying the evolution equations
\be
(i\gamma^\mu_j\partial_{x_j,\mu}-m_x)\psi(x_1...x_m)
= e_x\:\gamma^\rho_j\: A_\rho(x_j) \: \psi(x_1...x_m)
\ee 
for $j\in\{1,\ldots,m\}$. This equation agrees with \eqref{LPx} if we set
\be
\Psi^{(m,n)}_{\mu_1...\mu_n}(x_1...x_m,y_1...y_n) = \psi(x_1...x_m) \,\frac{1}{\sqrt{n!}}  A_{\mu_1}(y_1) \cdots A_{\mu_n}(y_n)
\ee
for every $n\in\NNN\cup\{0\}$, which means that $\Psi$ is the tensor product of $\psi$ with a coherent state for the photons. That is, the usual equations for representing an external electromagnetic field come out of \eqref{LPx} and \eqref{LPy} if we assume that the photons are disentangled from the electrons and in a coherent state, and drop the term for the emission of photons by the electrons belonging to the system (but accept that photons may get created by sources outside the system). This situation suggests further that the classical regime, in which the electromagnetic field can be treated classically, corresponds to a coherent state of the photons, and conversely it makes understandable why the classical limit of photons should correspond to a \emph{field} instead of \emph{particles}: even if, on the level of the full quantum theory, photons are particles with coordinates $y_k$, the classical regime is one in which their common wave function $A_\mu$ is crucial, not their configuration $(y_1...y_n)$. 

\item {\it Comparison to another model.} In \cite{pt:2013c}, a similar system of multi-time equations was considered, also with $x$-particles emitting and absorbing $y$-particles as in $x \leftrightarrows x+y$, but both $x$s and $y$s were assumed to be Dirac particles, see Eq.~(36) of \cite{pt:2013c}. An alternative way of arriving at the equations \eqref{LPx} and \eqref{LPy}, independent of the work of Landau and Peierls, consists of starting from the model of \cite{pt:2013c} and modifying it so as to give the $y$-particle spin 1 and mass 0: First, the index associated with $y_k$ should be a space-time index $\mu_k$ instead of (as in \cite{pt:2013c}) a Dirac spin index $s_k$; after all, the free evolution of the wave function of a single spin-1 mass-0 particle is given by the Maxwell equation \eqref{Maxwell} without sources, $J_\nu=0$. Second, the coefficient $g_s$ appearing in Eq.~(36a) of \cite{pt:2013c}, which according to Eq.~(38) of \cite{pt:2013c} can in general have the form $g_{sr'r}$, must now take the form $g^{\mu s'}_s$, and the obvious (and presumably only) Lorentz-invariant choice of such coefficients is the Dirac gamma matrix $\gamma^{\mu s'}_s$, possibly times a scalar factor that adjusts the strength of the emission and absorption terms; this last factor is naturally identified with the charge $e$ of the electron. With these adjustments, Eq.~(36a) of \cite{pt:2013c} literally becomes \eqref{LPx} above. Third, the Green's function used in Eq.~(36b) of \cite{pt:2013c} agrees with $\delta^3_\mu$ on spacelike surfaces through the origin up to a coefficient $g_s$ (cf.\ Eq.s (32) and (41) in \cite{pt:2013c}), so our use of $e_x\,\delta^3_\mu \gamma_j^\mu$ amounts to merely a different notation for $G$. Now Eq.~(36b) is of first order in time; since the 1-photon equation \eqref{Maxwell} is of second order, it is in that equation, applied to $y_k$, that the creation term of Eq.~(36b) should appear. We would thus arrive from (36b) at the equation
\begin{multline}\label{LPyPT}
\text{`` }2\partial_{y_k}^\mu\partial^{~}_{y_k,[\mu}\Psi^{(m,n)}_{\mu_k=\nu]}(x_1...x_m,y_1...y_n) =\\
\frac{e_x}{\sqrt{n}}\sum_{j=1}^m \delta^3_{\mu}(y_k-x_j) \: \gamma_j^\mu \:\Psi^{(m,n-1)}_{\widehat{\mu_k}}(x_1...x_m,y_1...y_{k-1},y_{k+1}...y_n)\,.\text{ ''}
\end{multline}
However, this equation does not make sense as an index $\nu$ that appears on the left-hand side is missing on the right-hand side. A simple and Lorentz-invariant modification of \eqref{LPyPT} that will fix this problem is to include another factor $\gamma_{j\nu}$, which brings us to \eqref{LPy}, except that it also allows the possibility that the $\gamma_{j\nu}$ might appear to the left of the $\gamma_j^\mu$. Reasons for putting $\gamma_{j\nu}$ to the right of $\gamma_j^\mu$ are provided by the considerations of Sections~\ref{sec:field} and \ref{sec:scp}.

\item {\it On the question of uniqueness of the multi-time equations.} In the setting of \cite{pt:2013c}, it was observed in Remark 6 in Section 2.2 of \cite{pt:2013c} that two different systems of equations can be equivalent on $\sS_{xy}$ although they would differ outside of $\sS_{xy}$; specifically, the alternative system, see Eq.~(30) of \cite{pt:2013c}, has no creation term in the $y_k$ equation but an additional creation term in the $x_j$ equation. This leads to the question whether in a similar way there exists a system of equations equivalent to \eqref{LPx} and \eqref{LPy} that differs in that the right-hand side of \eqref{LPy} is replaced by zero while a further term, a creation term related to the right-hand side of \eqref{LPy}, gets added to the right-hand side of \eqref{LPx}. Maybe there is, but we have not been able to identify such a system of equations; the difficulty is that \eqref{LPy} is of second order in time and \eqref{LPx} of first order, whereas in \cite{pt:2013c} all equations are of first order.

\item {\it Divergence of the source.} Since the Maxwell equation \eqref{Maxwell} can be solved only if 
\be\label{divJ0}
\partial_\nu J^\nu =0\,,
\ee
the question arises whether this is true of \eqref{source}. At points where $y_k \neq x_j$ for all $j$, \eqref{source} vanishes constantly, so \eqref{divJ0} is satisfied. Points where $y_k=x_j$ for some $j$, on the other hand, lie on the boundary of $\sS_{xy}$, so that $y_k$ cannot be varied in any timelike direction while $x_j$ is kept fixed; as a consequence, at such points, the Maxwell equations do not require an integrability condition of the form \eqref{divJ0}. It would be of interest to have a rigorous mathematical proof of this statement.

\item {\it External sources.} If it is desired to include external sources (i.e., external charges not governed by a quantum theory) in the model, with charge current given by a vector field $J^\ext_\nu$ on space-time $\MMM^4$, then \eqref{LPy} should be modified by adding
\be\label{extrasource}
\frac{1}{\sqrt{n}} J^\ext_\nu(y_k) \,\Psi^{(m,n-1)}_{\widehat{\mu_k}}(x_1...x_m,y_1...y_{k-1},y_{k+1}...y_n)
\ee
on the right-hand side. For example, for $m=0$ (i.e., pure radiation, no electrons) and $J^\ext_\nu$ corresponding to $M$ point charges $e_1,\ldots, e_M$ at rest in some Lorentz frame at locations $\vX_1,\ldots,\vX_M$,
\be
J^\ext_\nu(y^0,\vy) = \sum_{j=1}^M e_j \,\delta_{\nu 0} \, \delta^3(\vy-\vX_j)  \,,
\ee
a simple solution of \eqref{LPy} + \eqref{extrasource} is given by the coherent state
\be
\Psi^{(0,n)}_{\mu_1...\mu_n}(y_1...y_n) = \frac{1}{\sqrt{n!}} A_{\mu_1}(y_1) \cdots A_{\mu_n}(y_n)
\ee
with 
\be
A_\mu(y^0,\vy) = \Biggl( \sum_{j=1}^M \frac{e_j}{4\pi\, |\vy-\vX_j|}, \: \vzero\Biggr)\,.
\ee
To illustrate the physical effects of this photon wave function, we may consider a Dirac test particle, i.e., one that does not contribute a source term to \eqref{LPy} but is influenced by the photon wave function according to
\be
(i\gamma^\mu\partial_{x,\mu}-m_x)\psi(x) \, \Psi^{(0,n)}(y^{4n})=e_x\sqrt{n+1} \, \gamma^\rho \, \psi(x) \, \Psi^{(0,n+1)}_{\mu_{n+1}=\rho}(y^{4n},x)
\ee
in analogy to \eqref{LPx}: then the wave function $\psi$ of the test particle evolves as if it felt the Coulomb potential of the external charges. This consideration can be regarded as an outline of how the Coulomb potential arises from the (multi-time) Landau-Peierls model.

\item {\it Complex structure of vector potentials.} Since classical Maxwell fields have \emph{real} $F_{\mu\nu}$ and $A_\mu$, the question arises whether it is physically correct to consider \emph{complex} $F_{\mu\nu}$ and $A_\mu$. The Hilbert space of photon wave functions must be a complex Hilbert space, or else it would not be possible to form the complex tensor product with the Hilbert space of electron wave functions, as needed for entangled wave functions of electrons and photons together. The most obvious option is to consider complex $F_{\mu\nu}$, but one may think of the following alternative, so we point out what the problem is with that alternative. The space of all real $F_{\mu\nu}$ at a given point $x\in\MMM^4$ (sometimes denoted $\Lambda^2 T_x\MMM^4$), which has 6 real dimensions (that may be thought of as the 3 components of $\vE$ and the 3 components of $\vB$), also carries the structure of a complex vector space of complex dimension 3 by taking the multiplication by $i$ to be given by the Hodge $*$ operator, which maps $\Lambda^2 T_x\MMM^4$ to itself, is real-linear, and has the property $*^2=-1$. (Equivalently, the real $F_{\mu\nu}$ can be bijectively, real-linearly, and covariantly translated into symmetric 2-spinors of rank 2, $\phi_{AB}=\phi_{BA}$, which form a complex 3-dimensional vector space \cite[Eq.~(3.4.20)]{PR84}.) Since the $*$ operator is Lorentz covariant, so is the complex structure. It follows that the space $\sF$ of all real fields $F_{\mu\nu}(x)$ also naturally carries a complex structure. Therefore, if any \emph{real} linear gauge condition $\sA$ were given, i.e., a real subspace of the space of all real vector fields $A_\mu(x)$ with the property that $d:\sA\to\sF$ is bijective, then $d$ would naturally transport the complex structure to $\sA$, so $\sA$ would become a complex vector space and thus could perhaps serve as the Hilbert space of photon wave functions---although every element of $\sA$, when regarded as a function $A_\mu(x)$, has \emph{real} vector values. The problem with this space as the 1-photon space of $\Psi$ is that Eq.~\eqref{LPx} cannot be defined for it because elements of the tensor product Hilbert space (of, say, 1 electron and 1 photon) cannot be written as functions of two variables $x_1,y_1$. After all, for $f(x_1,y_1)$, multiplication by $i$ means literally multiplication of the value of the function $f$ at $(x_1,y_1)$, whereas elements of the tensor product space, with $i$ obtained from $*$, cannot be evaluated at a specific configuration $(x_1,y_1)$.
\end{enumerate}
\setcounter{remarks}{\theenumi}

\section{$A_\mu$ vs $F_{\mu\nu}$ Representations}
\label{sec:AF}

Since the wave function $\Psi^{(m,n)}$ plays the role of the vector potential in each $y_k$, it is a natural possibility to apply the exterior derivative $d=d_{\mu\nu}^\rho$ as in \eqref{ddef} and \eqref{FdA} to the variable $y_k$ (i.e., to use $y_k$-derivatives and the index $\rho=\mu_k$) to switch from ``vector potential representation'' to ``field tensor representation.'' Let us write $d_k$ for $d$ applied to $y_k$; e.g., the left-hand side of \eqref{LPy} equals $\partial^\mu_{y_k} (d_k\Psi)_{\mu\nu}$.

\subsection{Switching Between Representations}

The possibility to switch between the two representations is unique to multi-time wave functions because the $d$ operator involves taking a time derivative, and multi-time wave functions allow us to take a time derivative for just one particle but not others. For a single-time many-particle wave function, the only possible time derivative is one that affects all particles, i.e., not only the other photons but also the electrons, and that would not be physically appropriate for obtaining a field tensor representation of the photons. For example, consider a quantum state with a sharp number $m=1$ of electrons and a sharp number $n=1$ of photons, suppose that the two particles are widely separated and do not interact, and suppose that they are not entangled; in short, writing $x$ for $x_1$ etc.,
\be\label{ex11}
\Psi^{(1,1)}_{s\mu}(x,y) = \psi_s(x)\, A_\mu(y)
\ee
with $\psi$ a Dirac wave function. Then $(d_y\Psi^{(1,1)})_{s\mu\nu}(x,y) = \psi_s(x) \, F_{\mu\nu}(y)$, which is reasonable, whereas the single-time derivative would act on $\Psi^{(m,n)}$ like $\sum_{j=1}^m\partial_{x_j0}+\sum_{k=1}^n\partial_{y_k0}$ and thus yield for \eqref{ex11} an additional term of the form $[\partial_{x0}\psi_s(x)]\, A_\mu(y)$. We regard the possibility to switch between the two representations as an advantage of multi-time wave functions.

Note that $d_k$ and $d_\ell$ commute because, for $k\neq \ell$, they act on different variables and different indices. Since we can choose between the two representations for each photon separately, we may imagine $n$ ``switches''; flicking each of them downward to the field tensor representation leads to
\be
F^{(m,n)}_{s_1...s_m,\mu_1\nu_1...\mu_n\nu_n}(x_1...x_m,y_1...y_n) := d_1\cdots d_n \Psi^{(m,n)}(x_1...x_m,y_1...y_n)\,.
\ee
By virtue of \eqref{LPx} and \eqref{LPy}, $F$ satisfies the equations
\begin{subequations}
\begin{multline}\label{LPxF}
(i\gamma^\mu_j\partial_{x_j,\mu}-m_x)F^{(m,n)}(x_1...x_m,y_1...y_n)
=\\ 
e_x\sqrt{n+1} \:\gamma^\rho_j\:d_1\cdots d_n \Psi^{(m,n+1)}_{\mu_{n+1}=\rho}(x_1...x_m,y_1...y_n,x_j)
\end{multline}
and
\begin{multline}\label{LPyF}
\partial_{y_k}^{\mu_k} F^{(m,n)}_{\mu_k\nu_k}(x_1...x_m,y_1...y_n) =\\
\frac{e_x}{\sqrt{n}}\sum_{j=1}^m \delta^3_{\mu}(y_k-x_j) \: \gamma_j^\mu \gamma_{j\nu_k} \: F^{(m,n-1)}_{\widehat{\mu_k\nu_k}}(x_1...x_m,y_1...y_{k-1},y_{k+1}...y_n)\,.
\end{multline}
\end{subequations}
Expressed in a particular Lorentz frame in terms of $\vE$ and $\vB$ (referring to $\mu_k\nu_k$), the last equation reads
\begin{subequations}\label{17}
\begin{align}
\partial_{y_k,0} \vE_k^{(m,n)} &= \boldsymbol{\nabla}_{y_k} \times \vB_k^{(m,n)} - \frac{e_x}{\sqrt{n}} \sum_{j=1}^m \delta^3(\vy_k-\vx_j) \, \valpha_j \, F^{(m,n-1)}_{\widehat{\mu_k\nu_k}}(\widehat{y_k}) \label{17a}\\ 
\boldsymbol{\nabla}_{y_k} \cdot  \vE_k^{(m,n)} &= \frac{e_x}{\sqrt{n}} \sum_{j=1}^m \delta^3(\vy_k-\vx_j) \, F^{(m,n-1)}_{\widehat{\mu_k\nu_k}}(\widehat{y_k}).\label{17b}
\end{align}
In the same notation the other Maxwell equations, a consequence of $F_{\mu\nu}=2\partial_{[\mu}A_{\nu]}$, read
\begin{align}
\partial_{y_k,0} \vB_k^{(m,n)} &= -\boldsymbol{\nabla}_{y_k} \times \vE_k^{(m,n)} \label{17c}\\
\boldsymbol{\nabla}_{y_k} \cdot \vB_k^{(m,n)} &=0\,.\label{17d}
\end{align}
\end{subequations}
These equations are analogous to \eqref{MaxwellEB}.

\subsection{Uniqueness of Solutions}

For any Cauchy surface\footnote{Leaving out some technical details not relevant for our purposes, a Cauchy surface is a maximal spacelike 3d surface.} $\Sigma$, we denote the set of configurations on $\Sigma$ by
\begin{subequations}
\begin{align}
\Gamma_{xy}(\Sigma) 
&= \{ (x_1...x_m,y_1...y_n)\in \sS_{xy}: \text{all }x_j,y_k\in \Sigma \}\\
&= \bigcup_{m=0}^\infty \bigcup_{n=0}^\infty \Sigma^m \times \Sigma^n\,.
\end{align}
\end{subequations}
Generally, a multi-time wave function $\Psi$ on the set of spacelike configurations always defines for every Cauchy surface $\Sigma$ a wave function $\Psi_\Sigma$ (``surface wave function'') simply by restriction to configurations on $\Sigma$,
\be
\Psi_\Sigma = \Psi\Big|_{\Gamma_{xy}(\Sigma)} \,.
\ee

\begin{prop}
Given a linear gauge condition $\sA$ and initial data for $F^{(m,n)}$ for every $m,n$ on some Cauchy surface $\Sigma_0$, the solution $\Psi$ of \eqref{LPx}, \eqref{LPy}, and \eqref{LPg} is unique.
\end{prop}

For this proposition, we do not have a proof, but we have the following heuristic argument. (It would be of interest to have a proof.) Think of surface wave functions $\Psi_\Sigma$ and imagine that we change $\Sigma$ only in a small region by pushing it to the future and thus obtain another Cauchy surface $\Sigma'$. For a configuration $(x_1...x_m,y_1...y_n)$, this more or less means moving at most one point $x_j$ or $y_k$ to the future while keeping the others fixed. It is known that the solution of the Maxwell equation for a given source term $J_\nu$ is unique, and it was assumed that the gauge condition $\sA$ is such that ``turning the switch upward'' is an operation with a unique outcome. Furthermore, it is well known that the solution of the Dirac equation is unique. Thus, regardless of whether the one point was an electron or a photon, $\Psi_{\Sigma'}$ (and thus also $F_{\Sigma'}$) should be uniquely determined by $F_\Sigma$ and $\sA$. The same reasoning applies to local changes in $\Sigma$ toward the past. Combining local changes, we can go from any Cauchy surface to any other. Finally, if $\Psi_\Sigma$ is uniquely determined for every Cauchy surface $\Sigma$ then, since every spacelike configuration lies on some Cauchy surface, also $\Psi^{(m,n)}(x_1...x_m,y_1...y_n)$ is uniquely determined for every spacelike configuration.

The question of existence of solutions is closely related to the consistency of the multi-time equations and will be discussed in Section~\ref{sec:consistency}.

\section{Gauge Invariance}
\label{sec:gauge}

In this section, we describe how to change $\Psi$ when we want to change the gauge condition $\sA$ and note that the new wave function $\widetilde\Psi$ satisfies again the time evolution equations \eqref{LPx} and \eqref{LPy}. 

\subsection{Classical Gauge Transformations}

These transformations are in a way analogous to the traditional gauge transformations, which are defined as follows. They concern the transformation of a classical (real-valued) Maxwell field $A_\mu$, satisfying the Maxwell equation \eqref{Maxwell} with given (real-valued) source current vector field $J_\nu(y)$ with $\partial^\nu J_\nu=0$, according to
\be\label{AtildeA}
\widetilde{A}_\mu = A_\mu - \frac{1}{e_x}\partial_\mu \theta 
\ee 
for an arbitrary differentiable function $\theta:\MMM^4\to \RRR$ (and $e_x$ the electron charge) and the transformation of a 1-particle wave function $\psi:\MMM^4\to\CCC^4$ satisfying the Dirac equation \eqref{Dirac} according to 
\be\label{psitildepsi}
\widetilde{\psi}_s(x) = e^{i\theta(x)}\psi_s(x)
\ee
with the same function $\theta$. Under this joint transformation, $\widetilde{F}_{\mu\nu}=F_{\mu\nu}$, in particular the Maxwell equation \eqref{Maxwell} is still satisfied, and $\widetilde{\psi}$ satisfies the Dirac equation \eqref{Dirac} with $A_\mu$ replaced by $\widetilde{A}_\mu$.

\subsection{Multi-Time Wave Functions}

Now return to the multi-time wave functions and suppose we want to replace a gauge condition $\sA$ by a different one $\widetilde\sA$. If $d:\sA\to\sF$ and $\widetilde{d}:\widetilde\sA\to \sF$ are both the exterior derivative (but on different spaces), then $\widetilde{d}^{-1}d: \sA \to \widetilde{\sA}$ is the linear mapping that maps every vector field $A_\mu(x)\in\sA$ to the $\widetilde{A}_\mu(x)\in\widetilde{\sA}$ with the same $F_{\mu\nu}(x)$. It follows that $\widetilde{A}$ and $A$ are related according to \eqref{AtildeA} with complex-valued $\theta$, and the function $\theta$ depends linearly on $A$. Therefore, the gauge transformations relevant to us here are not the ones corresponding to a fixed function $\theta$ but to a fixed linear dependence of $\theta$ on $A$. We will first consider infinitesimal gauge transformations, corresponding to the case that the subspace $\widetilde\sA$ is infinitesimally close to the subspace $\sA$ within the space of all complex vector fields. In that case, $\theta$ will be an infinitesimally small function
\be\label{thetaTheta}
\theta(x) = (\Theta A)(x) \, ds\,,
\ee
where $ds$ is some infinitesimal quantity and $\Theta$ a linear operator from $\sA$ to the space of $\CCC$-valued functions on $\MMM^4$, defined by the relation
\be\label{gds}
A_\mu - \frac{1}{e_x} \partial_\mu (\Theta A) \, ds = \widetilde{d}^{-1} d A~~~\forall A\in \sA \,,
\ee
in fact uniquely up to two freedoms: (i)~multiplication of $\Theta$ by a constant $c>0$ while replacing $ds$ by $c^{-1}ds$, and (ii)~addition to $\Theta$ of any linear operator $C$ from $\sA$ to $\CCC$ (thought of as the space of constant functions on $\MMM^4$).

\begin{defn}
Under the infinitesimal gauge transformation given by \eqref{gds} and any choice of $C$ and $ds$ (or $c$), $\Psi$ transforms to
\begin{align}
&\widetilde{\Psi}^{(m,n)}(x_1...x_m,y_1...,y_n)\nonumber\\[2mm]
&\quad=\Psi^{(m,n)}(x_1...x_m,y_1...y_n) \nonumber\\[2mm]
&\qquad-\frac{1}{e_x}\sum_{k=1}^n \bigl[I^{\otimes(m+k-1)}\otimes \partial_{\mu_k}\Theta \otimes I^{\otimes(n-k)} \bigr]\Psi^{(m,n)}(x_1...x_m,y_1...y_n)~ds \nonumber\\
&\qquad+i\sqrt{n+1} \sum_{j=1}^m \Bigl( \bigl[I^{\otimes(m+n)}\otimes \Theta \bigr] \Psi^{(m,n+1)} \Bigr)(x_1...x_m,y_1...y_n,x_j)~ds\,.
\label{Psitilde}
\end{align}
\end{defn}

Note that different values of $c$ actually lead to the same $\widetilde{\Psi}$ because the factors $\Theta$ and $ds$ always appear together, while different choices of $C$ will change $\widetilde{\Psi}$ merely by an (irrelevant) phase factor.

\begin{prop}\label{prop:gaugeds}
If $\Psi$ satisfies \eqref{LPx}, \eqref{LPy}, and \eqref{LPg}, then $\widetilde{\Psi}$ satisfies \eqref{LPx}, \eqref{LPy}, and the analog of \eqref{LPg} for $\widetilde{\sA}$.
\end{prop}

We give the proof in Appendix~\ref{app:gaugeds}. Now we want to generalize from infinitesimal to finite changes of the gauge condition $\sA$. Since any path $\sA(s)$ in the space of all linear gauge conditions, parameterized by $s\in [0,1]$, can be regarded as a succession of infinitesimal changes of gauge, it is already defined how to transform $\Psi$ from $\sA(0)$ to $\sA(1)$. 

\begin{prop}\label{prop:gauge}
For suitable choice of $C$, this transformation depends only on the end points $\sA(0)$ and $\sA(1)$, not on the path between them.
\end{prop}

We give the proof in Appendix~\ref{app:pathindep}. Together, Propositions~\ref{prop:gaugeds} and \ref{prop:gauge} show that for any two linear gauge conditions $\sA,\widetilde\sA$, there exists a linear operator $G$ that maps any solution $\Psi$ of \eqref{LPx}, \eqref{LPy}, and \eqref{LPg} to a solution $\widetilde\Psi$ of \eqref{LPx}, \eqref{LPy}, and the analog of \eqref{LPg} for $\widetilde\sA$. We do not know whether there exists a closed formula for $G$.

\subsection{Remarks}
\label{sec:rem2}

\begin{enumerate}
\setcounter{enumi}{\theremarks}
\item\label{rem:exTheta} {\it Example of $\Theta$.} The Coulomb gauge condition \eqref{Coulomb} can equivalently be written as $\partial^\mu A_\mu=n^\nu n^\mu \partial_\nu A_\mu$, where $n^\mu=(1,0,0,0)$; more generally, $n^\mu$ is the timelike basis unit vector of the Lorentz frame in which \eqref{Coulomb} holds. Different choices of $n^\mu$ correspond to different gauge conditions, so let us consider an infinitesimal change in $n^\mu$, $\widetilde{n}^\mu=n^\mu + v^\mu ds$, and compute the corresponding $\Theta$ operator. Normalization of $\widetilde{n}^\mu$ requires that $n_\mu v^\mu=0$. From $\partial^\mu\widetilde{A}_\mu= \widetilde{n}^\nu \widetilde{n}^\mu \partial_\nu\widetilde{A}_\mu$ and $\widetilde{A}_\mu = A_\mu -\tfrac{1}{e_x} \partial_\mu \Theta A \, ds$, we obtain that
\be
(n^\nu n^\mu-g^{\nu\mu})\partial_\nu \partial_\mu \Theta A=e_x(v^\nu n^\mu+n^\nu v^\mu)\partial_\nu A_\mu\,.
\ee
If we choose coordinates such that $n^\mu=(1,0,0,0)$ and $v^\mu=(0,1,0,0)$, then this reads $\Delta \Theta A= e_x(\partial_1 A_0 + \partial_0 A_1)$; if we assume square integrability at every time, then the Laplacian $\Delta$ possesses a unique inverse $\Delta^{-1}$ (multiplication by $1/(k_1^2+k_2^2+k_3^2)$ in Fourier space), and
\be
\Theta A = e_x \Delta^{-1}(\partial_1 A_0 + \partial_0 A_1)\,.
\ee

\item {\it Domain of $\Theta$.} Our introduction of $\Theta$ was based on $\widetilde{d}^{-1}$, which requires that $A_\mu(x)$ is given for all $x\in\MMM^4$, whereas later in \eqref{Psitilde} we want to apply $\Theta$ to $\Psi$, which is given only for spacelike configurations. In the example of Remark~\ref{rem:exTheta}, the computation of $\Theta$ actually does not require knowledge of $A_\mu(x)$ on \emph{all} of $\MMM^4$ but just requires inverting the Laplacian, which can be done on the horizontal $\{x^0=\mathrm{const.}\}$ 3-planes (and perhaps on Cauchy surfaces). We leave open the question under which conditions $\Theta A$ is well defined for $A$ given on the set of points spacelike separated from a given spacelike configuration. 

\item {\it Gauge transformation of $F^{(m,n)}$.} The fact that in classical electrodynamics, $F_{\mu\nu}$ does not change under gauge transformations might have suggested that $F^{(m,n)}$ (the $\Psi$ with ``all switches down'') is also invariant under the gauge transformations of \eqref{Psitilde}. This is not so. (This fact is unsurprising when we keep in mind that $F^{(m,n)}$ also contains electron degrees of freedom corresponding to the $\psi$ in \eqref{psitildepsi}.) In fact, it easily follows from \eqref{Psitilde} that under an infinitesimal change \eqref{gds} of the gauge condition, $F^{(m,n)}$ transforms to
\begin{align}
&\widetilde{F}^{(m,n)}(x_1...x_m,y_1...y_n)\nonumber\\[2mm]
&\quad = F^{(m,n)}(x_1...x_m,y_1...y_n)\nonumber\\[2mm]
&\qquad + i\sqrt{n+1} \sum_{j=1}^m \Bigl( \bigl[I^{\otimes(m+n)}\otimes \Theta d^{-1} \bigr] F^{(m,n+1)} \Bigr)(x_1...x_m,y_1...y_n,x_j)~ds\,.\label{Fds}
\end{align}

\item {\it Affine gauge transformations.} A different analog of the classical 1-particle gauge transformations given by \eqref{AtildeA} and \eqref{psitildepsi} is based on fixing a function $\theta:\MMM^4\to\RRR$, thought of as being the function that appears in \eqref{AtildeA} and \eqref{psitildepsi}, and carrying out ``this'' transformation on every electron and photon variable in $\Psi$. A coherent interpretation of this idea is formulated in \eqref{affinegds} and \eqref{affineg} below. Note, however, that the 1-particle transformation \eqref{AtildeA} maps a linear subspace $\sA$ to an \emph{affine} subspace (as $0\in \sA$ gets mapped to $-\tfrac{1}{e_x}\partial_\mu \theta$, which in general is non-zero), with the consequence that $\Psi$ when transformed this way will no longer obey the third law \eqref{LPg} for any linear gauge condition. (And if we allowed affine subspaces for $\sA$, then the evolution of $\Psi$ would no longer be linear.)

For an infinitesimal transformation, i.e., assuming that $\theta(x)$ is infinitesimally small or replacing $\theta(x) \to \theta(x) \, ds$, the desired transform $\breve{\Psi}$ of $\Psi$ is
\begin{align}
&\breve{\Psi}^{(m,n)}(x_1...x_m,y_1...y_n) \nonumber\\[2mm]
&\quad =\Psi^{(m,n)}(x_1...x_m,y_1...y_n) \nonumber\\[2mm]
&\qquad + i\sum_{j=1}^m \theta(x_j) \, \Psi^{(m,n)}(x_1...x_m,y_1...y_n)  \, ds\nonumber\\
&\qquad -\frac{1}{e_x\sqrt{n}} \sum_{k=1}^n \partial_{\mu_k} \theta(y_k) \, \Psi^{(m,n-1)}_{\widehat{\mu_k}}(\widehat{y_k}) \, ds\,.
\label{affinegds}
\end{align}

\begin{prop}\label{prop:gauge1}
If $\Psi$ satisfies \eqref{LPx} and \eqref{LPy}, then so does $\breve{\Psi}$. Also, $\breve{\Psi}$ obeys the fermionic and bosonic permutation symmetries.
\end{prop}

We give the proof in Appendix~\ref{app:gauge1}. Here, we can also give an explicit formula for a finite (non-infinitesimal) transformation:\begin{multline}\label{affineg}
\breve{\Psi}^{(m,n)}(x_1...x_m,y_1...y_n) 
=\\
e^{i\sum\limits_{j=1}^m\theta(x_j)} \!\! \sum_{\cI\subseteq\{1...n\}} \!\! (-1)^{\#\cI}\sqrt{\frac{(n-\#\cI)!}{n!}}\biggl(\prod_{k\in\cI} \frac{\partial_{\mu_k}\theta(y_k)}{e_x}\biggr)\Psi^{(m,n-\#\cI)}_{\widehat{\mu_{\cI}}}\bigl(\widehat{y_{\cI}}\bigr). 
\end{multline}
Here, the sum is over all subsets of $\{1,\ldots,n\}$, which includes $\cI=\emptyset$, and thus has $2^n$ summands. In particular, the sum always has at least one summand, even for $n=0$; an empty product has the value 1; $\#\cI$ denotes the number of elements of $\cI$, $\widehat{\mu_\cI}$ the omission of all $\mu_k$ with $k\in\cI$, and $\widehat{y_{\cI}}$ the omission of all $y_k$ with $k\in\cI$.

\item\label{rem:indispensable} {\it The gauge condition is indispensable.} One might wonder why we need the third law \eqref{LPg} and do not just admit all solutions of \eqref{LPx} and \eqref{LPy}. The answer is that those would be too many. Here is an example illustrating this, which we owe to Lukas Nullmeier. Consider any solution $\Psi$ of \eqref{LPx} and \eqref{LPy} with sharp electron number $m\in\NNN$ (i.e., concentrated on the sectors with the given $m$).
Let $c:\RRR\to\RRR$ be a smooth increasing function with $c(t)=0$ for $t\leq 0$ and $c(1)=1$ for $t\geq 1$; define, in some fixed Lorentz frame,
\be
\theta(x):= e^{i\pi c(x^0)/m}
\ee
and consider the function
\be
\Psi':= \tfrac{1}{2}\Psi+\tfrac{1}{2}\breve{\Psi}\,.
\ee
By linearity and Proposition~\ref{prop:gauge1}, $\Psi'$ is a solution of \eqref{LPx} and \eqref{LPy} that obeys fermionic and bosonic permutation symmetry. For any $x\in\MMM^4$ with $x^0\leq 0$, $\theta(x)=1$ and $\partial_\mu\theta(x)=0$, while for any $x$ with $x^0\geq 1$, $\theta(x)=-1$ and $\partial_\mu\theta(x)=0$. Therefore, $\Psi'=\Psi$ on all configurations with all $x_j^0\leq 0$ and $y_k^0\leq 0$, while $\Psi'=0$ on all configurations with all $x_j^0\geq 1$ and $y_k^0\geq 1$. In particular, any initial data could evolve to the constant 0 function. Therefore, conversely, constant 0 initial data at time 0 could evolve to any other solution after time $1$, and any initial data could evolve to any other data at any other time. Thus, \eqref{LPx} and \eqref{LPy} alone would not define a meaningful time evolution.
\end{enumerate}
\setcounter{remarks}{\theenumi}

\section{Field Operators}
\label{sec:field}

Our basis for proposing the multi-time equations \eqref{LPx}, \eqref{LPy}, \eqref{LPg} is that they form natural multi-time equations for a system of two species of relativistic particles, one with spin $\tfrac{1}{2}$ and mass $m_x$ and one with spin 1 and mass 0, and in fact form a natural multi-time version of the Hamiltonian of Landau and Peierls. But there is another way of arriving at the equations \eqref{LPx}, \eqref{LPy}, \eqref{LPg}, which we describe in this section.

In some approaches to QED, one uses field operators, $\hat\psi_s(x)$ for the Dirac field and $\hat A_\mu(x)$ for the quantized electromagnetic field, acting on a Hilbert space $\Hilbert$. They are connected to the multi-time wave function $\Psi$ according to \cite{Nik10,pt:2013c}
\be\label{Psiop}
\Psi^{(m,n)}_{s_1...s_m,\mu_1...\mu_n}(x_1...x_m,y_1...y_n)=
\frac{1}{\sqrt{m!n!}}
\Bigl\langle \emptyset \Big| \hat\psi_{s_1}(x_1)\cdots \hat\psi_{s_m}(x_m) \hat A_{\mu_1}(y_1)\cdots \hat A_{\mu_n}(y_n) \Big|\varphi\Bigr\rangle
\ee
for $(x^{4m},y^{4n})\in\sS_{xy}$ with $y_k\neq y_\ell$ for all $k\neq \ell$, $\varphi\in\Hilbert$ the quantum state in the field operator picture and $|\emptyset\rangle$ the Fock vacuum state. (Again, we leave aside the issue of the Dirac sea.) We show in this section that the multi-time equations \eqref{LPx}, \eqref{LPy}, \eqref{LPgex} follow from equations that the field operators are usually assumed to obey together with a Coulomb gauge condition on $\hat A_{\mu}$.

The usual commutation relations for the field operators can be expressed, for $x-x'\in\sS_0^4$, as
\begin{subequations}\label{commute}
\begin{align}
\Bigl\{ \hat \psi_s(x), \hat \psi_{s'}(x') \Bigr\} 
&= 0 \label{commutepsi}\\
\Bigl\{ \overline{\hat \psi(x)}{}^s, \hat \psi_{s'}(x') \Bigr\} 
&= \gamma_{s'}^{\mu s} \: \delta^3_\mu(x-x') \label{commutepsibar}\\
\Bigl[ \hat A_\mu(x), \hat A_{\mu'}(x') \Bigr] 
&= 0\label{commuteA}\\
\Bigl[ \hat \psi_s(x), \hat A_{\mu'}(x') \Bigr] 
=\biggl[ \overline{\hat \psi(x)}{}^s, \hat A_{\mu'}(x') \biggr] 
&= 0\label{commutepsiA}
\end{align}
\end{subequations}
with $[\cdot,\cdot]$ the commutator, $\{\cdot,\cdot\}$ the anti-commutator, and
\be\label{psibardef}
\overline{\hat\psi}{}^s = \sum_{s'} \hat\psi_{s'}^\dagger \, \gamma^{0s}_{s'} \,.
\ee
Since we are leaving out positrons, our $\hat\psi(x)$ is just an electron annihilation operator, in particular
\be\label{psiop0}
\hat\psi_s(x)|\emptyset\rangle = 0
\ee
for all $x\in \MMM^4$. The field equations for the field operators read
\begin{subequations}\label{op}
\begin{align}
\sum_{s'}\Bigl(i\gamma^{\mu s'}_s \bigl[\partial_\mu+ie_x\hat A_\mu(x)\bigr] -m_x\Bigr)\hat \psi_{s'}(x)
&= 0
\label{psiop}\\
2\partial^\mu\partial_{[\mu} \hat A_{\nu]}(x)
&= e_x \, \overline{\hat\psi(x)} \, \gamma_\nu \, \hat\psi(x) \,.\label{Aop}
\end{align}
\end{subequations}
If we want to impose a Coulomb gauge condition on the field operators $\hat{A}_\mu$, it states that
\be
\sum_{\mu=1}^3 \partial^\mu \hat A_\mu(x)
=0\,. \label{gop}
\ee

\begin{prop}\label{prop:op}
Suppose that \eqref{commute}--\eqref{op} hold. Then $\Psi$ defined by \eqref{Psiop} satisfies the equations \eqref{LPx} and \eqref{LPy} on $\sS_{xy}$. If also \eqref{gop} holds, then $\Psi$ also satisfies \eqref{LPgex} on $\sS_{xy}$.
\end{prop}

We give the proof in Appendix~\ref{app:op}.

\section{Preservation of Scalar Product}
\label{sec:scp}

Our claim in this section is more or less that for any two Cauchy surfaces $\Sigma$, $\Sigma'$, the time evolution from $\Psi_{\Sigma}$ to $\Psi_{\Sigma'}$ is  unitary. ``More or less'' because we do not prove existence and uniqueness of solutions, and because the relevant scalar product for photon wave functions is not positive definite. (The physically correct scalar product should be positive definite, but we are anyway allowing negative energies and any choice of gauge.)

The scalar product $\scp{\Psi}{\Phi}_\Sigma$ on $\Sigma$ for our wave function $\Psi$ and another wave function $\Phi$ of the same kind is the Fock-space many-particle extension of the appropriate 1-electron and 1-photon scalar product. It is well known (e.g., \cite[Sec.~7.3.4]{Tum22}) that for 1-electron wave functions $\psi,\phi$, the scalar product reads
\be
\scp{\psi}{\phi}_{1x\Sigma} = \int_\Sigma V(d^3x) \: n^\nu(x) \: \overline{\psi}(x) \: \gamma_\nu \: \phi(x)
\ee
and is actually positive definite. The scalar product for 1-photon wave functions is based on the one used already by Landau and Peierls (and sometimes mentioned in connection with the quantization method of Gupta and Bleuler \cite{Gup50,Ble50,GB}); here we provide a version applicable in the position representation on any Cauchy surface $\Sigma$ and permitting also negative energy contributions and arbitrary gauge. This version, discussed in Appendix~\ref{app:GuptaBleuler}, reads as follows for a 1-photon wave function $A_\mu$ and another such wave function $B_\mu$:
\begin{subequations}\label{scp1y}
\begin{align}
\scp{A}{B}_{1y\Sigma} 
&= -i \int_\Sigma V(d^3x) \: n^\nu(x) ~\Bigl[ A_\mu^* \partial^{~}_\nu B^\mu-A_\nu^* \partial^{~}_\mu B^\mu - (\partial^{~}_\nu A_\mu^*) B^\mu+(\partial^\mu A_\mu^*)B^{~}_\nu \Bigr]\\
&=  \int_\Sigma V(d^3x) \: n^\nu(x) \: A_\mu^*  D^{\mu\mu'}_\nu B^{~}_{\mu'} 
\end{align}
\end{subequations}
with
\be\label{Dmumunudef}
D^{\mu\mu'}_\nu = 
-ig^{\mu\mu'} \overrightarrow{\partial}^{~}_{\!\!\nu} 
+i \delta^\mu_\nu \overrightarrow{\partial}^{\!\mu'} 
+ i\overleftarrow{\partial}^{~}_{\!\!\nu} \, g^{\mu\mu'}
- i\overleftarrow{\partial}^{\!\mu} \, \delta^{\mu'}_\nu \,,
\ee
where the arrow indicates whether the derivative operator acts to the left (i.e., on $A$) or to the right (i.e., on $B$).\footnote{If we restore the constants of nature, so Maxwell's inhomogeneous equation reads $\partial^\mu F_{\mu\nu}=\mu_0 \, J_\nu$, then $\scp{A}{B}_{1y\Sigma}$ should contain a factor of $\mu_0^{-1}$.} The full scalar product then reads explicitly:
\begin{align}
&\scp{\Psi}{\Phi}_\Sigma = \sum_{m=0}^\infty \sum_{n=0}^\infty \int_{\Sigma^{m+n}} \hspace{-6mm} V(d^3x_1)\cdots V(d^3x_m) \, V(d^3y_1)\cdots V(d^3y_n) \: \times\nonumber\\
&~~~~~~~~~~~~~\times \Biggl[\prod_{j=1}^m n^{\nu_j}(x_j)\Biggr] \Biggl[ \prod_{k=1}^n n^{\nu'_k}(y_k) \Biggr] \sum_{\substack{s_1...s_m\\s'_1...s'_m}}\biggl\{ \Psi^{(m,n)}_{s_1...s_m,\mu_1...\mu_n}(x_1...x_m,y_1...y_n)^* \: \times \nonumber\\
&~~~~~~~~~~~~~\times \Biggl[\prod_{j=1}^m (\gamma^0\gamma_{\nu_j})^{s'_j}_{s_j} \Biggr] \Biggl[\prod_{k=1}^n D^{\mu_k\mu'_k}_{y_k,\nu'_k} \Biggr] \Phi^{(m,n)}_{s'_1...s'_m,\mu'_1...\mu'_n}(x_1...x_m,y_1...y_n)\biggr\}\,, \label{scpdef}
\end{align}
where the range of $\overrightarrow{\partial}$ and $\overleftarrow{\partial}$ lies within the curly brackets.

\begin{prop}\label{prop:unitary}
The scalar product is invariant under the time evolution given by \eqref{LPx}, \eqref{LPy}, \eqref{LPg}. That is, for any two Cauchy surfaces $\Sigma,\Sigma'$ and any two solutions $\Psi,\Phi$ of \eqref{LPx}, \eqref{LPy}, and \eqref{LPg},
\be
\scp{\Psi}{\Phi}_\Sigma=\scp{\Psi}{\Phi}_{\Sigma'}\,.
\ee
\end{prop}

We give the proof in Appendix~\ref{app:unitary}. It would be of interest to study whether this statement remains true in curved space-time.

\section{Agreement with Landau and Peierls}
\label{sec:LP}

When all time variables of a multi-time wave function are set equal, one obtains the single-time wave function of the chosen Lorentz frame. In this section, we show that this step, applied to the multi-time wave function governed by \eqref{LPx}, \eqref{LPy}, and \eqref{LPgex}, yields essentially the wave function of Landau and Peierls. Here, ``essentially'' means that there are a few subtleties or caveats:
\begin{itemize}
\item Landau and Peierls considered only the Coulomb gauge \eqref{LPgex}, so we will also assume it in this section.
\item Instead of the $A_\mu$ representation of a photon wave function, they used the $F_{\mu\nu}$ representation. So in this section, we will also use this representation.
\item They mistakenly put an $\alpha$ in two places where it should be $\beta$.
\item They assumed a constraint condition, see \eqref{constraint} below, that we do not assume and that simplifies the expressions. We discuss this condition in Section~\ref{sec:constraint} and conclude that it is actually physically not reasonable.

\item They forgot the combinatorial factors $\sqrt{n}$ and $\sqrt{n+1}$ visible in \eqref{LPx} and \eqref{LPy}, but included an extra factor of 2 in certain places that (we argue) should not be there, or else their Hamiltonian is not Hermitian on Fock space.
\end{itemize}
In this section, we discuss these points and provide the single-time form of \eqref{LPx}, \eqref{LPy}, and \eqref{LPgex}; up the differences mentioned, it agrees with the equations of Landau and Peierls. The right-hand sides of \eqref{LPx} and \eqref{LPy} become the annihilation and creation terms, respectively, in the Landau-Peierls Hamiltonian.

\subsection{Single-Time Wave Function}
\label{sec:1time}

The wave function that Landau and Peierls used is the single-time version of $F$, i.e., $F$ restricted to the set
\be
\underset{=}{\sS}{}_{xy} = \Bigl\{ (x_1...x_m,y_1...y_n)\in\sS_{xy}: x_1^0=...=x_m^0=y_1^0=...=y_n^0 \Bigr\}
\ee
of simultaneous (horizontal, equal-time) configurations relative to a chosen Lorentz frame. We write
\be
\underset{=}{F} := F\Bigl|_{\underset{=}{\sS}{}_{xy}}
~~\text{and}~~
\underset{=}{\Psi} := \Psi\Bigl|_{\underset{=}{\sS}{}_{xy}}
\ee
for these ``equal-time'' wave functions. 
It is easy to derive the equation governing the single-time evolution of $\underset{=}{F}$ in $t=x_1^0=...=y_n^0$; we start from the chain rule
\be\label{1time1}
i\partial_t \underset{=}{F}^{(m,n)} = \Biggl[ \sum_{j=1}^m i\partial_{x_j0}F^{(m,n)} + \sum_{k=1}^n i\partial_{y_k0} F^{(m,n)} \Biggr]_{\underset{=}{\sS}{}_{xy}}
\ee
and express the right-hand side using \eqref{LPxF}, \eqref{LPyF}, and \eqref{17c} (or \eqref{LPxF}, \eqref{17a}, and \eqref{17c}):
\begin{align}
i\partial_t \underset{=}{F}^{(m,n)} &= \sum_{j=1}^m H^\free_{x_j} \underset{=}{F}^{(m,n)} 
+ \sum_{j=1}^m e_x\sqrt{n+1} \alpha_j^\rho \Bigl[d_1\cdots d_n \Psi^{(m,n+1)}_{\mu_{n+1}=\rho}(y_{n+1}=x_j)\Bigr]_{\underset{=}{\sS}{}_{xy}}\nonumber\\
&~~+ \sum_{k=1}^n H^\free_{y_k} \underset{=}{F}^{(m,n)} +\frac{e_x}{\sqrt{n}}\sum_{j=1}^m \sum_{k=1}^n \delta^3(\vy_k-\vx_j) \begin{pmatrix} 0&\valpha_j^T \\ -\valpha_j & 0_{3\times 3} \end{pmatrix}_{\! \mu_k\nu_k} \underset{=}{F}{}^{(m,n-1)}_{\widehat{\mu_k\nu_k}}(\widehat{y_k})\,, \label{1time2}
\end{align}
where $\alpha^\rho=\gamma^0\gamma^\rho$ (in particular, $\alpha^0 = I$), $T$ means transposed vector (row vector),
\be
H^\free_{x_j} = -i\valpha_j\cdot \boldsymbol{\nabla}_{x_j}+m_x\beta_j
\ee
is the free Dirac operator, and $H^\free_{y_k}$ means the 1-particle Hamiltonian corresponding to the evolution according to the source-free Maxwell equation acting on $y_k$ and $\mu_k,\nu_k$; that is, if we rewrite $F_{\mu\nu}$ as the pair $\bigl(\begin{smallmatrix}\vE\\ \vB\end{smallmatrix}\bigr)$,
\begin{align}
H^\free_y \begin{pmatrix}\vE\\ \vB\end{pmatrix} 
&= \begin{pmatrix} i \, \boldsymbol{\nabla} \times \vB\\ -i \, \boldsymbol{\nabla} \times  \vE \end{pmatrix} \,.
\end{align}
In addition, $\underset{=}{F}$ satisfies the constraint conditions \eqref{17b} and \eqref{17d} or, in a different notation (where summation over $i=1,2,3$ is understood),
\begin{subequations}\label{constraint1}
\begin{align}
&\partial_{y_k}^i \underset{=}{F}{}^{(m,n)}_{\mu_k=0,\nu_k=i}
=\frac{e_x}{\sqrt{n}}\sum_{j=1}^m \delta^3(\vy_k-\vx_j)\, \underset{=}{F}{}^{(m,n-1)}_{\widehat{\mu_k\nu_k}}(\widehat{y_k}) \label{constraint1a}\\
&\partial_{y_k}^1 \underset{=}{F}{}^{(m,n)}_{\mu_k=2,\nu_k=3}
+\partial_{y_k}^2 \underset{=}{F}{}^{(m,n)}_{\mu_k=3,\nu_k=1}
+\partial_{y_k}^3 \underset{=}{F}{}^{(m,n)}_{\mu_k=1,\nu_k=2} =0\,. \label{constraint1b}
\end{align}
\end{subequations}
In order to close the single-time evolution equation \eqref{1time2}, i.e., in order to express $\Psi$ in terms of $F$, we use the relation \eqref{d-1ex}:
\begin{align}
i\partial_t \underset{=}{F}^{(m,n)} &= \sum_{j=1}^m H^\free_{x_j} \underset{=}{F}^{(m,n)} + \sum_{k=1}^n H^\free_{y_k} \underset{=}{F}^{(m,n)}  \nonumber\\
&~~+ \sum_{j=1}^m e_x\sqrt{n+1} \begin{pmatrix} (-\Delta_{y_{n+1}})^{-1} \boldsymbol{\nabla}_{y_{n+1}} \cdot \\ \valpha_j \cdot (-\Delta_{y_{n+1}})^{-1} \boldsymbol{\nabla}_{y_{n+1}}  \times \end{pmatrix}  \underset{=}{F}^{(m,n+1)}(y_{n+1}=x_j) \nonumber\\
&~~+ \frac{e_x}{\sqrt{n}}\sum_{j=1}^m \sum_{k=1}^n \delta^3(\vy_k-\vx_j)\, \begin{pmatrix} \valpha_j \\ \vzero \end{pmatrix} \, \underset{=}{F}{}^{(m,n-1)}_{\widehat{\mu_k\nu_k}}(\widehat{y_k})\,, \label{1time3}
\end{align}
using the notation $\bigl(\begin{smallmatrix}\vE\\ \vB\end{smallmatrix}\bigr)$ for the $y_{n+1}$-components of $F^{(m,n+1)}$ and the $y_k$ of $F^{(m,n)}$. Together, \eqref{constraint1} and \eqref{1time3} form the single-time equations to compare to those of Landau and Peierls.

\subsection{Landau and Peierls' Constraint Condition}
\label{sec:constraint}

Landau and Peierls used, as the wave function of a single photon, a complex $F_{\mu\nu}$ (equivalently expressed as complex $\vE$ and $\vB$), but imposed the constraint condition \cite[(5) and (22)]{LP30}
\be\label{constraint}
\vB= -i(-\Delta)^{-1/2}\boldsymbol{\nabla} \times \vE\,.
\ee

\begin{prop}\label{prop:constraint}
The relation \eqref{constraint} is true for complex solutions of the source-free ($J_\nu=0$) Maxwell equation with exclusively positive energies, but neither true nor preserved in time if $\boldsymbol{\nabla}\times\vJ\neq 0$.
\end{prop}

The proof is given in Appendix~\ref{app:constraint}. It is indeed possible that $\boldsymbol{\nabla}\times\vJ\neq 0$, as $J_\nu$ can be chosen in the Maxwell equation arbitrarily as long as $\partial^\nu J_\nu=0$. We therefore argue that it is not physically reasonable to impose \eqref{constraint}. For Landau and Peierls, it was possible to make some choices when defining their Hamiltonian (within the limits that the Hamiltonian is self-adjoint and the appropriate version of the Maxwell constraints \eqref{20b}, \eqref{20d} are preserved), and this particular choice does not seem physically correct, as it means a change of the evolution of 1-photon wave functions, away from the Maxwell equation. We will therefore henceforth drop \eqref{constraint} in our comparison with Landau and Peierls' equations.

\subsection{Comparison}

Eq.~(26) in their paper \cite{LP30} states the Schr\"odinger equation and corresponds to our \eqref{1time3}.
Their Eq.~(26a) states the analog of the Maxwell constraint $\boldsymbol{\nabla}\cdot \vE = J_0$ and corresponds to our \eqref{constraint1a} or, equivalently, \eqref{17b}. The equations agree up to the differences discussed above.

\section{Consistency}
\label{sec:consistency}

Multi-time equations can be inconsistent (see, e.g., \cite{LPT20}). In order for a set of multi-time equations to be acceptable as part of a physical theory, we need that it defines a multi-time evolution, so we need that it is consistent. For the equations \eqref{LPx}, \eqref{LPy}, \eqref{LPg}, we do not have a consistency proof. However, in this section we provide a discussion of what the issues are and of how a consistency proof might proceed.

\subsection{Known Arguments}
\label{sec:known}

Before we start, it is worth mentioning that Dirac, Fock, and Podolsky \cite{dfp:1932} proposed a multi-time formulation of a simplified version of QED and considered the consistency question for it. A key difference between their equations and the ones considered here is that the photons were not represented as particles but as a field, that is, the wave function was not taken to be a function of photon positions. Correspondingly, the wave function did not involve a separate time variable for each photon, but only $m+1$ time variables, one for each of a fixed number $m$ of electrons and one for the electromagnetic field. The consistency question was answered in the positive, on a heuristic level by Bloch \cite{bloch:1934} and on a rigorous level by Nickel and Deckert \cite{ND:2019,Nic2019} after implementing an ultraviolet cut-off.

A system of multi-time equations more similar to \eqref{LPx} and \eqref{LPy}, modeling $x$-particles emitting and absorbing $y$-particles, but assuming that both the $x$s and the $y$s are Dirac particles, was developed in \cite{pt:2013c}, and heuristic arguments for its consistency were given there. These arguments did not constitute a rigorous proof because the evolution equations were ultraviolet divergent and therefore ill-defined already in their single-time version. However, in \cite{LNT20} the arguments were turned into a rigorous consistency proof of a variant of the equations with ultraviolet cut-off, which underlines that the lack of rigor did not originate from the nature of the arguments in \cite{pt:2013c}. These arguments would apply to \eqref{LPx} and \eqref{LPy} as well if the equations were of first order and propagation local. \emph{Propagation locality} means that wave functions cannot propagate faster than at the speed of light $c$. Now it is well known that the Maxwell equation entails that $F_{\mu\nu}$ propagates no faster than at $c$, but whether this also applies to $A_\mu$ might depend on the gauge condition.

So, the situation seems to be as follows. If a gauge condition effectively provides a first-order equation for (the aspect of $\Psi$ corresponding to) $A_\mu$ and ensures propagation locality, then the arguments of \cite{pt:2013c} presumably entail consistency of the evolution equations. If the evolution is not propagation local, then it seems unlikely that the evolution equations could be consistent. If the evolution equations are effectively of second order, even after imposing a gauge, then further analysis seems necessary.

\subsection{Interaction Locality}

We call a time evolution law (single-time or multi-time) \emph{interaction local} if there are no interaction terms in the time evolution law for the wave function between spacelike separated regions. This seems to be the case for \eqref{LPx}, \eqref{LPy}, \eqref{LPg} and variants of it with different gauge conditions. That is relevant because of the following general pattern:

\begin{conj}
Suppose we are given a Hamiltonian $H$ such that interaction locality and propagation locality hold. Then there exists a unique consistent multi-time evolution on the set of spacelike configurations that satisfies interaction locality and propagation locality and agrees with the evolution defined by $H$ on the horizontal (i.e., simultaneous) configurations in the given Lorentz frame. If $H$ is Poincar\'e invariant then so is the multi-time evolution.
\end{conj}

This conjecture can be formulated as a precise mathematical statement. We will not provide such a formulation here but indicate briefly in Appendix~\ref{app:conj} how to obtain it. 

We believe that Conjecture 1 is correct. Of course, it does not directly apply to our case because (i)~it is not clear what the position PVM for photons should be; and (ii)~as in Section~\ref{sec:known}, it is not clear which gauge condition would effectively provide first-order evolution equations and ensure propagation locality, as assumed in Conjecture 1. It thus remains an open problem to find such a gauge condition.

\section{Interior-Boundary Condition}

The use of Dirac delta functions in the creation and annihilation terms of the time evolution equations often leads to ultraviolet divergence. In non-relativistic theories with particle creation and annihilation, this divergence could be tamed through the use of conditions on the wave function at configurations at which an $x$-particle and a $y$-particle meet, so-called \emph{interior-boundary conditions} (IBCs) \cite{TT15a,TT15b}. In this paper, we do not aim to address the ultraviolet divergence of the evolution equations \eqref{LPx}, \eqref{LPy}, \eqref{LPg}, and it is not clear whether this divergence can be ameliorated by using IBCs, but it may be worth pointing out that the equations \eqref{LPx}, \eqref{LPy}, \eqref{LPg} naturally lead to an IBC.

It is well known that the 8 Maxwell equations comprise 6 time evolution equations and 2 constraints,
\be
\boldsymbol{\nabla} \cdot  \vB=0~~\text{and}~~\boldsymbol{\nabla} \cdot  \vE=J_0\,.
\ee
For $J_0(\vy) = e_x \sum_{j=1}^m \delta^3(\vy-\vx_j)$, we obtain for small $r>0$ that
\be
e_x
=\int_{B_r(\vx_j)} \hspace{-5mm} d^3\vy\, \boldsymbol{\nabla} \cdot \vE(\vy)
= r^2\int_{\SSS^2} \hspace{-1mm} d^2\vomega\: \vomega\cdot \vE(\vx_j+r\vomega)
\ee
by the Ostrogradski--Gauss integral theorem. The same consideration applied to \eqref{LPy}, followed by taking $r\to 0$, leads to an IBC:
\begin{multline}
\lim_{r\to 0} \int_{\SSS^2}\!\! d^2\vomega \, r^2 \sum_{i=1}^3 \omega^i \partial_{y_{n+1},[i}^{~}\Psi_{\mu_{n+1}=0]}^{(m,n+1)}(x_1...x_m,y_1...y_n,(x_j^0,\vx_j+r\vomega)) =\\ 
\frac{e_x}{\sqrt{n}}\, \Psi^{(m,n)}(x_1...x_m,y_1...y_n).\label{IBC}
\end{multline}
Now the idea of the IBC method is to take care of the $\delta$ functions in the evolution equations for $\Psi$ by imposing \eqref{IBC} as a condition on $\Psi$. We finally mention that IBCs were studied in relativistic space-time in \cite{HT20,HPT24}.

\section{Conclusions}

We have presented Lorentz-invariant multi-time Schr\"odinger equations governing the evolution of the wave function for a model involving the emission and absorption of photons by electrons. When only simultaneous configurations (relative to a fixed Lorentz frame) are considered, the equations reduce to those proposed by Landau and Peierls \cite{LP30}. The multi-time formulation makes the Lorentz invariance and gauge freedom manifest and leads to simpler equations; in particular, while the Hamiltonian of Landau and Peierls contains pseudo-differential operators, the multi-time equations presented here do not---they are just PDEs. We have outlined several ways of arriving at the multi-time equations and verified that they have the expected properties of gauge invariance and preservation of the scalar product. For future work, it would be of interest to give a mathematical proof of their consistency, and to extend them to electron-positron pair creation and to other gauge theories.

It is a striking feature of the new equations that the photon absorption term plays the role of the connection coefficient of the covariant derivative in a vector bundle. Put differently, the equations in a way unify the concepts of connection coefficient and absorption term. In particular, the case of an external electromagnetic field is included as a photon wave function that is not entangled with the electrons.

The multi-time equations provide an unusual formulation of QED, a particle-position representation of QED. We believe that it is desirable to have several formulations of a theory available. It would be of interest to see how far the particle-position approach can be taken, and whether the hypothesis is viable that quantum field theory is essentially a relativistic kind of quantum mechanics with particle creation and annihilation. 

An intriguing aspect, in our humble opinion, of this approach is what it suggests about the nature of the wave function. It is often emphasized that the quantum state can be represented mathematically in many ways (such as the Schr\"odinger, Heisenberg, interaction, or Tomonaga-Schwinger picture), possibly very abstractly as just a vector in some abstract Hilbert space, if our goal is to formulate an algorithm for computing empirical predictions. On the other hand, if we allow ourselves to think of the wave function itself as a physical object, then it is natural to ask what kind of object it is. In view of the equations discussed here, it seems like a serious and natural possibility that the wave function is actually a function on $\sS_{xy}$.

\appendix

\section{Construction of 3d Dirac delta Distribution}
\label{app:delta}

Here are three arguments to the effect that there exists a distribution $D$ characterized by \eqref{characterizeD}.

\begin{enumerate}
\item We define at first a distribution $\tilde D$ on the larger set $\MMM^4$ and then obtain $D$ as the restriction of $\tilde D$ to $\sS_0^4$. We set
\be\label{tildeDdef}
\tilde D^{\rho}(x^0,\vx) = \bigl(\delta^3(\vx),0,0,0\bigr)
\ee
in a chosen Lorentz frame; $\tilde D$ is not Lorentz invariant. It can be integrated over any smooth, spacelike Cauchy surface $\Sigma$ in $\MMM^4$ against a smooth function $f:\Sigma\to \CCC$ to yield (regarding $\Sigma$ as the graph of the function $x^0=h(\vx)$)
\begin{multline}\label{tildeD}
\int_\Sigma V(d^3x) \, n_\rho(x) \, \tilde D^\rho(x) \, f(x) ~~\stackrel{x=(h(\vx),\vx)}{=} \\
\int_{\RRR^3} \! d^3\vx \, \sqrt{-\det {}^3g} \: n_0(\vx)\: \delta^3(\vx) \: f(h(\vx),\vx) = f(h(\vzero),\vzero)
\end{multline}
because the push-forward of the standard basis vectors in $\RRR^3$ under $(h,\mathrm{id}):\RRR^3\to\Sigma$ is $v_1=(\partial_1 h, 1,0,0)$ etc., so
\be
{}^3g_{ij} = v_i^\mu v_j^\nu g_{\mu\nu}  = 
\begin{cases}
(\partial_i h)^2-1 &\text{if }i=j\\
(\partial_i h)(\partial_j h) &\text{if } i\neq j,
\end{cases}
\ee
which has determinant
\be
\det {}^3g = (\partial_1h)^2+(\partial_2h)^2+(\partial_3h)^2-1\,,
\ee
while 
\be
n^\mu = \frac{1}{\sqrt{1-|\nabla h|^2}}(1,\partial_1h,\partial_2h,\partial_3h)
\ee
(as that is a future-timelike unit vector orthogonal to $v_1,v_2,v_3$); as a consequence,
\be
\sqrt{-\det {}^3g} \: n_0(\vx) = 1.
\ee
To restrict $\tilde D$ to $\sS_0^4$ means to consider only surfaces $\Sigma$ passing through $0\in\MMM^4$, so $h(\vzero)=0$, and \eqref{tildeD} yields \eqref{characterizeD}, which proves the existence of a distribution $D$ with the property \eqref{characterizeD}.

\item The following alternative argument also starts from $\tilde D$ as in \eqref{tildeDdef} but uses distributions $\tilde D'$ on $\MMM^4$ obtained by applying a Lorentz transformation $\Lambda$ to $\tilde D$. It also shows that $\tilde D$ and $\tilde D'$ agree on $\sS_0^4$; thus, $D$ could also be defined as the restriction of $\tilde D'$ for any $\Lambda$ to $\sS_0^4$.

If $\Lambda$ maps $(1,0,0,0)$ to the (future-timelike unit) vector $u^\mu=(u^0,\vu)$, then it maps $\tilde D$ to $\tilde D'$ given by
\be
\tilde D^{\prime\rho}(x^0,\vx) = \delta^3(\vx-\tfrac{x^0}{u^0}\vu)\: \tfrac{1}{u^0} \: u^\rho\,.
\ee
As a consequence, for $\Sigma$ passing through 0, $\int_\Sigma \tilde D \: f$ can be evaluated as follows. Let $\Lambda$ be a Lorentz transformation that (keeps 0 fixed and) maps $n^\mu(0)$ to $(1,0,0,0)$; thus, $\Lambda$ maps $\Sigma$ to a surface $\Sigma'$ that is tangent at 0 to the horizontal 3-plane $P_0$ through 0, and $f$ to $f'=f\circ \Lambda$, so
\be
\int_\Sigma \tilde D \: f = \int_{\Sigma'} \tilde D' \: f'\,.
\ee
Suppose that $f$ is defined, not only on $\Sigma$, but on all of $\sS_0^4$. Since $\tilde D'$ vanishes on $\Sigma'$ away from 0, we can replace $\Sigma'$ by $P_0$, i.e.,
\be
\int_{\Sigma'} \tilde D' \: f' = \int_{P_0} \tilde D' \: f'\,.
\ee
But on $P_0$, it is straightforward to evaluate the integral, yielding
\be\label{Dprimef}
\int_{P_0} \tilde D' \: f' =\tfrac{u^0}{u^0}\, f\circ \Lambda(0) =   f(0) \,.
\ee
This value is independent of how $f$ was extended from $\Sigma$ to $\sS_0^4$, as it should be. It follows that $D$, defined as the restriction of $\tilde D$ to $\sS_0^4$, obeys \eqref{characterizeD}.

Moreover, since also for $\tilde D''$ obtained from $\tilde D$ through any other Lorentz transformation,
\be
\int_{P_0} \tilde D'' \: f' =\tfrac{(u'')^0}{(u'')^0}\, f\circ \Lambda(0) =   f(0) \,,
\ee
i.e., \eqref{Dprimef} with $\tilde D'$ replaced by $\tilde D''$ yields the same value, any Lorentz transform of $\tilde D$ will yield the same value for fixed $f$ and $\Sigma$ through 0, so all Lorentz transforms of $\tilde D$ coincide on $\sS_0^4$.

\item The propagator (or Green's function or time evolution kernel) $S$ of the free Dirac equation \cite[p.~15]{thaller:1992} is a distribution on $\MMM^4$ with values in the complex $4\times 4$ matrices (or endomorphisms of spin space) that can be defined through the property that the solution $\psi:\MMM^4\to\CCC^4$ of the free Dirac equation with initial data $\psi_0:\RRR^3\to\CCC^4$ at $x^0=0$ has the form
\be\label{psiSpsi}
\psi(x) = \int_{\RRR^3} \! d^3\vy \: S(x-(0,\vy)) \:\gamma^0 \: \psi_{0}(\vy)\,.
\ee
(Note that the $S$ in \cite{thaller:1992} differs from our $S$ by a factor of $\gamma^0$.) Equivalently, the $s'$-th column of the matrix $S$ is the solution of the Dirac equation with initial condition given by the $s'$-th column of $\gamma^0$ times the 3d Dirac delta,
\be\label{Sdef}
S_s^{~s'}(0,\vx)= \gamma_s^{0s'} \: \delta^3(\vx)\,.
\ee
It is known that $S$ is Lorentz invariant and vanishes at all $x$ with $x\spacelike 0$. 

Now consider $\int_\Sigma S \, \gamma^\mu f$ for $\Sigma$ through 0 and $f:\Sigma \to\CCC$. Let $\Lambda$ be again a Lorentz transformation that (keeps 0 fixed and) maps $n_\mu(0)$ to $(1,0,0,0)$ and $f'=f\circ \Lambda$. 
Then, since $S$ is Lorentz invariant and $\gamma^\mu$ is Lorentz invariant,
\be
\int_\Sigma S\,\gamma^\mu f = \int_{\Sigma'} S \,\gamma^\mu\, f' \,.
\ee
Since $S$ vanishes at $x\spacelike 0$, we can replace $\Sigma'$ by $P_0$,
\be
\int_{\Sigma'} S\,\gamma^\mu \, f' = \int_{P_0} S\,\gamma^\mu \, f'\,.
\ee
But on $P_0$, we can evaluate the integral using \eqref{Sdef} to obtain, with $I$ the identity endomorphism,
\be
\int_{P_0} S\, \gamma^\mu \, f' = \int_{\RRR^3}\!\! d^3\vx \: \gamma^0  \gamma^0  \, \delta^3(\vx) \, f'(0,\vx) = I \: f\circ \Lambda(0) = I\: f(0)\,.
\ee
That is, the Lorentz invariant distribution $S \gamma^\mu$ yields, for every $f:\Sigma\to\CCC$, $If(0)$. Since $I$ is Lorentz invariant, it follows that there is a Lorentz invariant distribution $D^\mu$ such that $S \gamma^\mu = I \, D^\mu$, and $D^\mu$ satisfies \eqref{characterizeD}.
\end{enumerate}

\section{Proof of Proposition~\ref{prop:gaugeds}: Gauge Invariance}
\label{app:gaugeds}

\begin{proof}
We repeat the definition \eqref{Psitilde}:
\begin{align}
\widetilde{\Psi}^{(m,n)}
&=\Psi^{(m,n)}\nonumber\\[2mm]
&\quad -\frac{1}{e_x}\sum_{k=1}^n  \partial_{y_k,\mu_k}\Theta_k \Psi^{(m,n)}~ds \nonumber\\
&\quad +i\sqrt{n+1} \sum_{i=1}^m \Bigl( \Theta_{n+1} \Psi^{(m,n+1)} \Bigr)(y_{n+1}=x_i)~ds\,, \label{Psitilde2}
\end{align}
where we have renamed the electron index $j$ into $i$ for later convenience (not to be confused with the imaginary unit), used the notation $\Theta_k$ for the $\Theta$ operator acting on $y_k$, and mentioned variables only when necessary.

To check \eqref{LPx} for $\widetilde\Psi$, we consider the left-hand side of \eqref{LPx} for $\widetilde\Psi$, insert \eqref{Psitilde2}, and calculate, using that operations acting on different variables commute:
\begin{subequations}
\begin{align}
&(i\gamma_j^\mu \partial_{x_j,\mu}-m_x) \widetilde{\Psi}^{(m,n)}
= (i\gamma_j^\mu \partial_{x_j,\mu}-m_x) \Psi^{(m,n)} \nonumber\\[2mm]
&\quad - \frac{1}{e_x} \sum_{k=1}^n \partial_{y_k,\mu_k}\Theta_k (i\gamma_j^\mu\partial_{x_j,\mu}-m_x) \Psi^{(m,n)}\, ds\nonumber\\
&\quad +i\sqrt{n+1} (i\gamma_j^\mu\partial_{x_j,\mu}-m_x) \sum_{i\neq j} \Bigl[\Theta_{n+1}\Psi^{(m,n+1)} \Bigr](y_{n+1}=x_i) \, ds\nonumber\\
&\quad +i\sqrt{n+1}(i\gamma_j^\mu\partial_{x_j,\mu}-m_x)\Bigl(\Bigl[\Theta_{n+1}\Psi^{(m,n+1)}\Bigr](y_{n+1}=x_j)\Bigr) \, ds\\
\intertext{[now use \eqref{LPx} for $\Psi$]}
&= e_x\sqrt{n+1}\gamma_j^{\mu_{n+1}}\Psi^{(m,n+1)}(y_{n+1}=x_j) \nonumber\\[2mm]
&\quad -\sqrt{n+1}\gamma_j^{\mu_{n+1}} \sum_{k=1}^n \Bigl[\partial_{y_k,\mu_k} \Theta_k \Psi^{(m,n+1)} \Bigr](y_{n+1}=x_j) \, ds\nonumber\\
&\quad +ie_x\sqrt{(n+1)(n+2)}\sum_{i\neq j} \Bigl(\Theta_{n+1}\gamma_j^{\mu_{n+2}}\Psi^{(m,n+2)}_{\mu_{n+2}}(y_{n+2}=x_j)\Bigr)(y_{n+1}=x_i) \, ds \nonumber\\
&\quad +ie_x\sqrt{(n+1)(n+2)}\Bigl(\Theta_{n+1}\gamma_j^{\mu_{n+2}}\Psi^{(m,n+2)}_{\mu_{n+2}}(y_{n+2}=x_j)\Bigr)(y_{n+1}=x_j) \, ds \nonumber\\[2mm]
&\quad -\sqrt{n+1} \gamma_j^{\mu_{n+1}} \Bigl(\partial_{y_{n+1},\mu_{n+1}} \Theta_{n+1} \Psi^{(m,n+1)} \Bigr)(y_{n+1}=x_j) \, ds\\
\intertext{[now use bosonic symmetry and reorder terms]}
&= e_x\sqrt{n+1}\gamma_j^{\mu_{n+1}}\Psi^{(m,n+1)}(y_{n+1}=x_j) \nonumber\\[2mm]
&\quad -\sqrt{n+1}\gamma_j^{\mu_{n+1}} \sum_{k=1}^{n+1} \Bigl[\partial_{y_k,\mu_k} \Theta_k \Psi^{(m,n+1)} \Bigr](y_{n+1}=x_j) \, ds\nonumber\\
&\quad +ie_x\sqrt{(n+1)(n+2)}\sum_{i=1}^m \Bigl(\Theta_{n+2}\gamma_j^{\mu_{n+1}}\Psi^{(m,n+2)}_{\mu_{n+1}}\Bigr)(y_{n+1}=x_j,y_{n+2}=x_i) \, ds \\
\intertext{[now use \eqref{Psitilde2} backward]}
&=e_x \sqrt{n+1} \gamma_j^{\mu_{n+1}} \widetilde{\Psi}^{(m,n+1)}(y_{n+1}=x_j)
\end{align}
\end{subequations}
as claimed.

Condition~\eqref{LPg} is satisfied by construction. In order to check \eqref{LPy} for $\widetilde{\Psi}$, it suffices to check \eqref{LPyF} for $\widetilde{F}$ given by \eqref{Fds}. Starting from the left-hand side of \eqref{LPyF} for $\widetilde{F}$ and using \eqref{Fds},
\begin{subequations}
\begin{align}
&\partial_{y_k}^{\mu_k} \widetilde{F}^{(m,n)}_{\mu_k\nu_k}
=\partial_{y_k}^{\mu_k} F^{(m,n)}_{\mu_k\nu_k}\nonumber\\[2mm]
&\quad +i\sqrt{n+1}\sum_{j=1}^m \Bigl( \Theta_{n+1} d_{n+1}^{-1} \partial_{y_k}^{\mu_k} F^{(m,n+1)}_{\mu_k\nu_k} \Bigr)(y_{n+1}=x_j) \, ds\\
\intertext{[now use \eqref{LPyF} for $F$]}
&= \frac{e_x}{\sqrt{n}} \sum_{j=1}^m \delta^3_\mu(y_k-x_j)\, \gamma_j^\mu \, \gamma_{j\nu_k} \, F^{(m,n-1)}_{\widehat{\mu_k\nu_k}}(\widehat{y_k}) \nonumber\\
&\quad +ie_x\sum_{i,j=1}^m \Bigl( \Theta_{n+1} d_{n+1}^{-1} \delta^3_\mu(y_k-x_i) \, \gamma_i^\mu \, \gamma_{i\nu_k} F^{(m,n)}_{\mu_k\nu_k}(\widehat{y_k},y_{n+1}) \Bigr)(y_{n+1}=x_j) \, ds\\
\intertext{[now use \eqref{Fds} backward]}
&= \frac{e_x}{\sqrt{n}} \sum_{j=1}^m \delta^3_\mu(y_k-x_j) \, \gamma_j^\mu \, \gamma_{j\nu_k} \widetilde{F}^{(m,n-1)}_{\widehat{\mu_k\nu_k}}(\widehat{y_k})
\end{align}
\end{subequations}
as claimed.
\end{proof}

\section{Proof of Proposition~\ref{prop:gauge}: Path Independence of Gauge Transform}
\label{app:pathindep}

\noindent{\it Proof.}
Abstractly, the rule \eqref{Psitilde} defines a connection (parallel transport) in the bundle with fibers $\sA$ over the space $\sG$ of all linear gauge conditions. Path independence is equivalent to the vanishing of the curvature of this connection. Equivalently, we need to show for an infinitesimal parallelogram in $\sG$ (see Figure~\ref{fig:parallelogram}) with side length of order $ds$ that parallel transport along the path $\sA_1\sA_2\sA_4$ agrees up to second order in $ds$ with parallel transport along $\sA_1\sA_3\sA_4$.

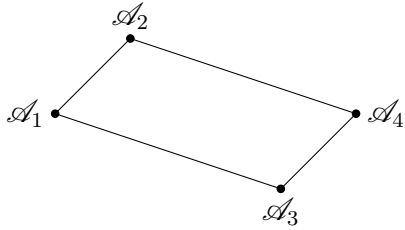
\begin{figure}[h]
\begin{center}
\begin{tikzpicture}
  \filldraw (0,0) circle [radius=0.05];
  \filldraw (1,1) circle [radius=0.05];
  \filldraw (3,-1) circle [radius=0.05];
  \filldraw (4,0) circle [radius=0.05];
  \draw (0,0) -- (1,1);
  \draw (1,1) -- (4,0);
  \draw (0,0) -- (3,-1);
  \draw (3,-1) -- (4,0);
  \node at (-0.4,0) {$\sA_1$};
  \node at (1,1.3) {$\sA_2$};
  \node at (3,-1.3) {$\sA_3$};
  \node at (4.4,0) {$\sA_4$};
\end{tikzpicture}
\end{center}
\caption{Infinitesimal parallelogram in the space $\sG$ of all linear gauge conditions}
	\label{fig:parallelogram}
\end{figure}

\subsection{Classical Maxwell Fields}

To this end, we note first that for the classical Maxwell field, the parallel transport of the field from $\sA_i$ to $\sA_j$ is indeed path-independent because it is given by $d_j^{-1}d_i$. Likewise, the corresponding transformation \eqref{psitildepsi} of $\psi$ is path-independent up to a global (i.e., $x$-independent) phase factor, which can be asssumed to be 1 without loss of generality (thereby narrowing the choice of $C$).

We now derive some formulas from this path-independence. Let $A_1=A_{1\mu}(x)$ be any element of $\sA_1$, let $A_i$ denote the result of transporting to $\sA_i$, let $v_{ij}ds$ be the vector in $\sG$ from $\sA_j$ to $\sA_i$, and let $\Theta_{ij}$ be the $\Theta$ operator for transport from $\sA_j$ to $\sA_i$, so
\be\label{A2A1first}
A_2 = A_1 -\tfrac{1}{e_x} \partial (\Theta_{21} A_1)ds + O(ds^2)\,,
\ee
where $\partial$ means the gradient of a function on space-time.
However, we need to keep track of terms up to second order. To this end, we regard $\Theta_{21}$ as a special case of a mapping that can be applied to any tangent vector $v$ to $\sG$ at $\sA_1$ and yields the transport operator along $v\, ds$, i.e.,
\be
\Theta_{i1} = \Theta v_{i1} \,.
\ee
Next we can ask how $\Theta$ changes as we move the starting point of the transport away from $\sA_1$; let $D\Theta$ be the derivative of $\Theta$ at $\sA_1$ and $D_{j1}:= v_{j1}\cdot D$ the directional derivative in the direction of $v_{j1}$. Then (see Figure~\ref{fig:parallelogram})
\begin{subequations}\label{Thetaijds}
\begin{align}
\Theta_{43} &= \Theta_{21}+ D_{31}\Theta_{21} ds+O(ds^2)\\
\Theta_{42} &= \Theta_{31}+ D_{21}\Theta_{31} ds+O(ds^2)\,.
\end{align}
\end{subequations}
The second-order extension of the expansion \eqref{A2A1first} then reads
\be\label{A2A1second}
A_2 = A_1 -\tfrac{1}{e_x} \partial \Theta_{21} A_1ds + \tfrac{1}{2e_x^2} (\partial \Theta_{21} \circ \partial\Theta_{21}-e_x \partial D_{21} \Theta_{21}) A_1 ds^2+O(ds^3)\,.
\ee
(Note that $\partial$ is the operator that takes the derivative in space-time, not to be confused with $D$, the derivative in $\sG$.) This follows from expanding the transport along $v_{21}ds$, given by $\sT \exp\bigl(-\tfrac{1}{e_x}\int_0^1\Theta_{21}(\sA_1+v_{21}t\, ds)dt\bigr)$ with $\sT \exp$ the time-ordered exponential, into a Dyson series up to second order in $ds$.

The relation \eqref{A2A1second} remains true when replacing $A_2$ by $A_3$ and $\Theta_{21}$ by $\Theta_{31}$. Now $A_4$ can be obtained from either $A_2$ or $A_3$, and we know it is path independent, so, on the one hand,
\begin{align}
A_4 &= A_1 -\tfrac{1}{e_x}\partial(\Theta_{42}+\Theta_{21})A_1 \, ds \nonumber\\
&\quad + \tfrac{1}{2e_x^2}\partial\Bigl(\Theta_{42}\circ \partial\Theta_{42}-e_x  D_{42}\Theta_{42} \nonumber\\
&~~~~~~~~~~~ + \Theta_{21}\circ \partial \Theta_{21} - e_x D_{21} \Theta_{21} \nonumber\\
&~~~~~~~~~~~ + 2\Theta_{42}\circ\partial \Theta_{21}\Bigr) A_1 \, ds^2 + O(ds^3)
\end{align} 
and, on the other hand, the same relation holds after replacing $42\to 43$ and $21\to 31$.

Taking the difference and noting that two functions on space-time with equal gradient differ only by addition of a constant, we obtain the relation
\begin{align}
&-\tfrac{1}{e_x}(\Theta_{42}+\Theta_{21}) ds \nonumber\\
&\quad + \tfrac{1}{2e_x^2}\Bigl(\Theta_{42}\circ \partial\Theta_{42}-e_x  D_{42}\Theta_{42} \nonumber\\
&~~~~~~~~~ + \Theta_{21}\circ \partial \Theta_{21} - e_x D_{21} \Theta_{21} \nonumber\\
&~~~~~~~~~ + 2\Theta_{42}\circ\partial \Theta_{21}\Bigr) ds^2 \nonumber\\
&= \text{(the same with $42\to43$ and $21\to31$)} + C'ds + O(ds^3)\,, 
\label{basicindep}
\end{align} 
where $C'$ is a linear operator $\sA_1\to \CCC$ (where $\CCC$ is thought of as the constant functions on $\MMM^4$) that vanishes under the above assumption that the transformation \eqref{psitildepsi} is path independent.

\subsection{Wave Functions}

We now use the formulas \eqref{Thetaijds} and \eqref{basicindep} to derive the path independence of the transformation $\widetilde\Psi$ of $\Psi$. Along the path $\sA_1\sA_2\sA_4$, we obtain that, up to $O(ds^3)$,
\begin{subequations}
\begin{align}
\Psi_4^{(m,n)}
&=\Psi_1^{(m,n)} \label{a}\\
&\quad-\tfrac{1}{e_x} \sum_k^n \partial_{\mu_k} \Theta_{21,k} \Psi_1^{(m,n)}ds \label{b}\\
&\quad+i\sqrt{n+1}\sum_j^m \Theta_{21,n+1} \Psi_1^{(m,n+1)}(y_{n+1}=x_j)ds \label{c}\\
&\quad+\tfrac{1}{2e_x^2} \sum_{k\neq k'}^n \partial_{\mu_k} \Theta_{21,k} \partial_{\mu_{k'}} \Theta_{21,k'} \Psi_1^{(m,n)}ds^2 \label{d}\\
&\quad+\tfrac{1}{2e_x^2} \sum_k^n (\partial_{\mu_k} \Theta_{21,k} \circ \partial\Theta_{21,k}-e_x\partial D_{21}\Theta_{21,k}) \Psi_1^{(m,n)}ds^2 \label{e}\\
&\quad-\tfrac{1}{2}\sqrt{(n+1)(n+2)} \sum_{j, j'}^m \Theta_{21,n+1}\Theta_{21,n+2}\Psi_1^{(m,n+2)}(y_{n+1}=x_j,y_{n+2}=x_{j'})ds^2 \label{f}\\
&\quad-\tfrac{i}{e_x}\sqrt{n+1} \sum_j^m \sum_k^n \partial_{\mu_k} \Theta_{21,k}\Theta_{21,n+1}\Psi_1^{(m,n+1)}(y_{n+1}=x_j)ds^2 \label{g}\\
&\quad-\tfrac{i}{2e_x}\sqrt{n+1} \sum_j^m \Bigr( (\Theta_{21,n+1}\circ \partial\Theta_{21,n+1} -e_x D_{21}\Theta_{21,n+1}) \Psi_1^{(m,n+1)}\Bigr) (y_{n+1}=x_j)ds^2\label{h}\\
\intertext{+ the same terms with 42 instead of 21}
&\quad + \tfrac{1}{e_x^2} \sum_{k\neq k'}^n \partial_{\mu_k}\Theta_{42,k} \partial_{\mu_{k'}}\Theta_{21,k'} \Psi_1^{(m,n)}ds^2 \label{i}\\
&\quad + \tfrac{1}{e_x^2} \sum_k^n \partial_{\mu_k} \Theta_{42,k}\circ \partial \Theta_{21,k} \Psi_1^{(m,n)}ds^2 \label{j}\\
&\quad-\tfrac{i}{e_x}\sqrt{n+1}\sum_j^m\sum_k^n \Bigl( \Theta_{42,n+1} \partial_{\mu_k} \Theta_{21,k} \Psi_1^{(m,n+1)} \Bigr) (y_{n+1}=x_j)ds^2 \label{k}\\
&\quad-\tfrac{i}{e_x}\sqrt{n+1}\sum_j^m \Bigl( \Theta_{42,n+1} \circ \partial \Theta_{21,n+1} \Psi_1^{(m,n+1)} \Bigr) (y_{n+1}=x_j)ds^2 \label{l}\\
&\quad-\tfrac{i}{e_x} \sqrt{n+1} \sum_j^m \sum_k^n \Bigl( \partial_{\mu_k} \Theta_{42,k} \Theta_{21,n+1} \Psi_1^{(m,n+1)}\Bigr) (y_{n+1}=x_j)ds^2 \label{m}\\
&\quad-\sqrt{(n+1)(n+2)}\sum_{j,j'}^m \Bigl( \Theta_{42,n+1} \Theta_{21,n+2} \Psi_1^{(m,n+2)} \Bigr)(y_{n+1}=x_j,y_{n+2}=x_{j'})ds^2. \label{n}
\end{align}
\end{subequations}
We now show that the index 2 can be replaced everywhere by 3: in \eqref{b} + \eqref{e} + \eqref{j} + the mirror terms ($21\leftrightarrow 42$) of \eqref{b} and \eqref{e} because of \eqref{basicindep}; likewise in \eqref{c} + \eqref{h} + \eqref{l} + the mirror terms of \eqref{c} and \eqref{h} because of \eqref{basicindep}; 
in \eqref{d} and \eqref{g}, each together with their mirror terms because of \eqref{Thetaijds}, as the correction terms caused be the replacement ($2\to3$) are of order $ds$ or higher and thus lead to a contribution $O(ds^3)$; 
in \eqref{i} and \eqref{f}, each together with their mirror terms, and in \eqref{k} and \eqref{m} together because of \eqref{Thetaijds}, using that the two $\Theta$'s act on different variables; in \eqref{n} because of \eqref{Thetaijds}, bosonic symmetry, and again that the two $\Theta$'s act on different variables.
\qed

\section{Proof of Proposition~\ref{prop:gauge1}: Affine Gauge Transform}
\label{app:gauge1}

\begin{proof}
We repeat the definition \eqref{affinegds}:
\begin{align}
\breve{\Psi}^{(m,n)} &=\Psi^{(m,n)} \nonumber\\[2mm]
&\quad + i\sum_{i=1}^m \theta(x_i) \, \Psi^{(m,n)}  \, ds\nonumber\\
&\quad -\frac{1}{e_x\sqrt{n}} \sum_{\ell=1}^n \partial_{\mu_\ell} \theta(y_\ell) \, \Psi^{(m,n-1)}_{\widehat{\mu_\ell}}(\widehat{y_\ell}) \, ds\,,
\label{affinegds2}
\end{align}
where we have renamed $j$ into $i$ and $k$ into $\ell$ for later convenience.

To check \eqref{LPx} for $\breve{\Psi}$, we start from the left-hand side of \eqref{LPx} for $\breve{\Psi}$ and insert the definition \eqref{affinegds2}:
\begin{subequations}
\begin{align}
&(i\gamma_j^\mu \partial_{x_j,\mu}-m_x) \breve{\Psi}^{(m,n)}
= (i\gamma_j^\mu \partial_{x_j,\mu}-m_x) \Psi^{(m,n)} \nonumber\\[2mm]
&\quad + i\sum_{i\neq j} \theta(x_i) \, (i\gamma_j^\mu \partial_{x_j,\mu}-m_x)\Psi^{(m,n)}  \, ds\nonumber\\
&\quad - \gamma_j^\mu \bigl(\partial_{\mu}\theta(x_j)\bigr) \, \Psi^{(m,n)}  \, ds\nonumber\\[2mm]
&\quad +i \theta(x_j) \, (i\gamma_j^\mu \partial_{x_j,\mu}-m_x)\Psi^{(m,n)}  \, ds\nonumber\\
&\quad -\frac{1}{e_x\sqrt{n}} \sum_{\ell=1}^n \partial_{\mu_\ell} \theta(y_\ell) \, (i\gamma_j^\mu \partial_{x_j,\mu}-m_x) \Psi^{(m,n-1)}_{\widehat{\mu_\ell}}(\widehat{y_\ell}) \, ds\\
\intertext{[now use \eqref{LPx} for $\Psi$]}
&= e_x\sqrt{n+1} \gamma_j^{\mu_{n+1}} \Psi^{(m,n+1)}_{\mu_{n+1}}(y_{n+1}=x_j) \nonumber\\[2mm]
&\quad + i\sum_{i\neq j} \theta(x_i) \, e_x\sqrt{n+1} \gamma_j^{\mu_{n+1}} \Psi^{(m,n+1)}_{\mu_{n+1}}(y_{n+1}=x_j)  \, ds\nonumber\\
&\quad - \gamma_j^\mu \bigl(\partial_{\mu}\theta\bigr)(x_j) \, \Psi^{(m,n)}  \, ds\nonumber\\[2mm]
&\quad +i \theta(x_j) \, e_x\sqrt{n+1} \gamma_j^{\mu_{n+1}} \Psi^{(m,n+1)}_{\mu_{n+1}}(y_{n+1}=x_j)  \, ds\nonumber\\
&\quad - \sum_{\ell=1}^n \partial_{\mu_\ell} \theta(y_\ell) \, \gamma_j^{\mu_{n+1}} \Psi^{(m,n)}_{\widehat{\mu_\ell},\mu_{n+1}}(\widehat{y_\ell},y_{n+1}=x_j) \, ds\\
\intertext{[now use \eqref{affinegds2} backward]}
&= e_x \sqrt{n+1} \gamma_j^{\mu_{n+1}} \, \breve{\Psi}^{(m,n+1)}(y_{n+1}=x_j)
\end{align}
\end{subequations}
as claimed.

To check \eqref{LPy} for $\breve{\Psi}$,
\begin{subequations}
\begin{align}
&2\partial_{y_k}^\mu\partial^{~}_{y_k,[\mu}\breve{\Psi}^{(m,n)}_{\mu_k=\nu]}
=2\partial_{y_k}^\mu\partial^{~}_{y_k,[\mu} \Psi^{(m,n)}_{\mu_k=\nu]} \nonumber\\[2mm]
&\quad + i\sum_{i=1}^m \theta(x_i) \, 2\partial_{y_k}^\mu\partial^{~}_{y_k,[\mu}\Psi^{(m,n)}_{\mu_k=\nu]}  \, ds\nonumber\\
&\quad -\frac{1}{e_x\sqrt{n}} \sum_{\ell=1,\ell\neq k}^n \partial_{\mu_\ell} \theta(y_\ell) \ 2\partial_{y_k}^\mu\partial^{~}_{y_k,[\mu}\Psi^{(m,n-1)}_{\mu_k=\nu], \widehat{\mu_\ell}}(\widehat{y_\ell}) \, ds \nonumber\\
&\quad - \frac{1}{e_x\sqrt{n}} \Bigl[ 2\partial_{y_k}^\mu \underbrace{\partial^{~}_{y_k,[\mu} \partial_{y_k,\nu]}}_{=0} \theta(y_k)  \Bigr] \Psi^{(m,n-1)}_{\widehat{\mu_k}}(\widehat{y_k}) \, ds\\
\intertext{[now use \eqref{LPy} for $\Psi$]}
&=\frac{e_x}{\sqrt{n}}\sum_{j=1}^m \delta^3_{\mu}(y_k-x_j) \: \gamma_j^\mu \gamma_{j\nu} \:\Psi^{(m,n-1)}_{\widehat{\mu_k}}(\widehat{y_k}) \nonumber\\[2mm]
&\quad + i\sum_{i=1}^m \theta(x_i) \, \frac{e_x}{\sqrt{n}}\sum_{j=1}^m \delta^3_{\mu}(y_k-x_j) \: \gamma_j^\mu \gamma_{j\nu} \:\Psi^{(m,n-1)}_{\widehat{\mu_k}}(\widehat{y_k})  \, ds\nonumber\\
&\quad -\frac{1}{\sqrt{n(n-1)}} \sum_{\ell=1,\ell\neq k}^n \partial_{\mu_\ell} \theta(y_\ell) \sum_{j=1}^m \delta^3_{\mu}(y_k-x_j) \: \gamma_j^\mu \gamma_{j\nu} \:\Psi^{(m,n-2)}_{\widehat{\mu_k\mu_\ell}}(\widehat{y_k,y_\ell}) \, ds\\
\intertext{[now use \eqref{affinegds2} backward]}
&=\frac{e_x}{\sqrt{n}} \sum_{j=1}^m \delta^3_\mu(y_k-x_j) \, \gamma_j^\mu \, \gamma_{j\nu} \: \breve{\Psi}^{(m,n-1)}_{\widehat{\mu_k}}(\widehat{y_k})
\end{align}
\end{subequations}
as claimed.

The statement about the fermionic and bosonic permutation symmetries can be verfied directly from the definition \eqref{affinegds}.
\end{proof}

\section{Proof of Proposition~\ref{prop:op}: Field Operators}
\label{app:op}

\begin{proof}
In order to derive \eqref{LPx}, apply $(i\gamma^\mu_j\partial_{x_j,\mu}-m_x)$ to \eqref{Psiop} to obtain
\begin{subequations}
\begin{align}
&\sqrt{m!n!}(i\gamma^\mu_j\partial_{x_j,\mu}-m_x)\Psi^{(m,n)}=\nonumber\\[2mm]
& \sum_{s'_j}
\Bigl\langle \emptyset \Big| \hat\psi_{s_1}(x_1)\cdots \Bigl[(i\gamma^{\mu s'_j}_{s_j}\partial_{x_j,\mu}-m_x)
\hat\psi_{s'_j}(x_j)\Bigr] \cdots \hat\psi_{s_m}(x_m) \hat A_{\mu_1}(y_1)\cdots \hat A_{\mu_n}(y_n) \Big|\varphi\Bigr\rangle\\
&\stackrel{\eqref{psiop}}{=} \sum_{s'_j}
\Bigl\langle \emptyset \Big| \hat\psi_{s_1}(x_1)\cdots 
\Bigl[e_x \gamma^{\mu s'_j}_{s_j} \hat A_\mu(x_j) \,
\hat\psi_{s'_j}(x_j)\Bigr] \cdots \hat\psi_{s_m}(x_m) \hat A_{\mu_1}(y_1)\cdots \hat A_{\mu_n}(y_n) \Big|\varphi\Bigr\rangle\\
&\hspace{-4mm}\stackrel{\eqref{commutepsiA},\eqref{commuteA}}{=} e_x\sum_{s'_j}\gamma^{\mu s'_j}_{s_j}
\Bigl\langle \emptyset \Big| \hat\psi_{s_1}(x_1)\cdots
\hat\psi_{s'_j}(x_j) \cdots \hat\psi_{s_m}(x_m) \hat A_{\mu_1}(y_1)\cdots \hat A_{\mu_n}(y_n) \, \hat A_\mu(x_j) \Big|\varphi\Bigr\rangle\\
&\stackrel{\eqref{Psiop}}{=}e_x\sum_{s'_j}\gamma^{\mu s'_j}_{s_j}\sqrt{m!(n+1)!} \Psi^{(m,n+1)}_{\mu_{n+1}=\mu}(x_1...x_m,y_1...y_n,x_j)\,.
\end{align}
\end{subequations}

Likewise, in order to derive \eqref{LPy}, evaluate the left-hand side of \eqref{LPy} using \eqref{Psiop} to obtain
\begin{subequations}
\begin{align}
&\sqrt{m!n!}2\partial_{y_k}^\mu\partial^{~}_{y_k,[\mu}\Psi^{(m,n)}_{\mu_k=\nu]}=\nonumber\\
&2\partial_{y_k}^\mu\partial^{~}_{y_k,[\mu}
\Bigl\langle \emptyset \Big| \hat\psi_{s_1}(x_1)\cdots \hat\psi_{s_m}(x_m) \hat A_{\mu_1}(y_1)\cdots \hat A_{\mu_k=\nu]}(y_k) \cdots \hat A_{\mu_n}(y_n) \Big|\varphi\Bigr\rangle\\
&\stackrel{\eqref{Aop}}{=} e_x
\Bigl\langle \emptyset \Big| \hat\psi_{s_1}(x_1)\cdots \hat\psi_{s_m}(x_m) \hat A_{\mu_1}(y_1)\cdots 
\overline{\hat\psi(y_k)} \gamma_\nu \hat\psi(y_k)
\cdots \hat A_{\mu_n}(y_n) \Big|\varphi\Bigr\rangle\\
&\stackrel{\eqref{commutepsiA}}{=} e_x
\Bigl\langle \emptyset \Big| \hat\psi_{s_1}(x_1)\cdots \hat\psi_{s_m}(x_m) \overline{\hat\psi(y_k)} \gamma_\nu \hat\psi(y_k)
\hat A_{\mu_1}(y_1)\cdots \widehat{\hat A(y_k)} \cdots
\hat A_{\mu_n}(y_n) \Big|\varphi\Bigr\rangle\,,\label{LPyfromop1}
\end{align}
where $\widehat{\hat A(y_k)}$ means that the factor $\hat A(y_k)$ is omitted.
Since, by \eqref{commutepsi} and \eqref{commutepsibar},
\be
\Bigl[\hat\psi_s(x) ,\overline{\hat\psi(y_k)} \gamma_\nu 
\hat\psi(y_k) \Bigr] = \sum_{s',s''}\gamma_s^{\mu s'} \, \delta^3_\mu(y_k-x) \, \gamma_{\nu s'}^{~~s''} \hat\psi_{s''}(y_k)\,,
\ee
we can transform the last expression \eqref{LPyfromop1} further to
\begin{align}
&e_x
\Bigl\langle \emptyset \Big| \overline{\hat\psi(y_k)} \gamma_\nu \hat\psi(y_k)\hat\psi_{s_1}(x_1)\cdots \hat\psi_{s_m}(x_m) 
\hat A_{\mu_1}(y_1)\cdots \widehat{\hat A_{\mu_k}(y_k)} \cdots
\hat A_{\mu_n}(y_n) \Big|\varphi\Bigr\rangle \nonumber\\
&+ \sum_{j=1}^m e_x  \, \delta^3_\mu(y_k-x_j) \, \gamma_j^\mu\gamma_{j\nu} 
\Bigl\langle \emptyset \Big| \hat\psi_{s_1}(x_1)\cdots\hat\psi_{s_m}(x_m) 
\hat A_{\mu_1}(y_1)\cdots \widehat{\hat A_{\mu_k}(y_k)} \cdots
\hat A_{\mu_n}(y_n) \Big|\varphi\Bigr\rangle  \,.\label{LPyfromop2}
\end{align}
By \eqref{psibardef} and \eqref{psiop0},
\be
\langle\emptyset| \overline{\hat\psi(y)}{}^s
= \sum_{s'}\langle\emptyset|\hat\psi_{s'}^\dagger(y) \, \gamma_{s'}^{0s}
= \sum_{s'}\langle\hat\psi_{s'}(y)\emptyset| \gamma_{s'}^{0s} =0\,,
\ee
so the first term in \eqref{LPyfromop2} vanishes, and the second term yields, by \eqref{Psiop},
\be
\sum_{j=1}^m e_x  \, \delta^3_\mu(y_k-x_j) \, \gamma_j^\mu\gamma_{j\nu} \sqrt{m!(n-1)!} \: \Psi^{(m,n-1)}_{\widehat{\mu_k}}(x_1...x_m,y_1...y_{k-1},y_{k+1}...y_n)\,,
\ee
\end{subequations}
in agreement with \eqref{LPy}.

Finally, in order to derive \eqref{LPgex}, apply $\sum_{\mu_k=1}^3 \partial_{y_k}^{\mu_k}$ to \eqref{Psiop} to obtain
\begin{multline}
\sqrt{m!n!}\sum_{\mu_k=1}^3 \partial_{y_k}^{\mu_k}\Psi^{(m,n)}_{s_1...s_m,\mu_1...\mu_n}(x_1...x_m,y_1...y_n)=\\
\sum_{\mu_k=1}^3 \partial_{y_k}^{\mu_k}\Bigl\langle \emptyset \Big| \hat\psi_{s_1}(x_1)\cdots \hat\psi_{s_m}(x_m) \hat A_{\mu_1}(y_1)\cdots \hat A_{\mu_n}(y_n) \Big|\varphi\Bigr\rangle
\stackrel{\eqref{gop}}{=}0\,.
\end{multline}
This completes the proof of Proposition~\ref{prop:op}.
\end{proof}

\section{Scalar Product of 1-Photon Wave Functions}
\label{app:GuptaBleuler}

The scalar product of photon wave functions, perhaps first described by Landau and Peierls \cite{LP30} in the Coulomb gauge in a fixed Lorentz frame, but also well known in a covariant form applicable to any gauge consistent with the Lorenz gauge condition, is generalized here further to arbitrary gauge and arbitrary Cauchy surfaces $\Sigma$. This general form, given already in \eqref{scp1y}, reads:
\begin{align}
\scp{A}{B}_{1y\Sigma} 
&= -i \int_\Sigma V(d^3x) \: n^\nu(x) ~\Bigl[ A_\mu^* \partial^{~}_\nu B^\mu-A_\nu^* \partial^{~}_\mu B^\mu - (\partial^{~}_\nu A_\mu^*) B^\mu+(\partial^\mu A_\mu^*)B^{~}_\nu \Bigr]\,. \label{scp1yc}
\end{align}
We will describe why it yields the known forms in the appropriate special cases. But first we note its properties.

\subsection{Properties}

\begin{prop}\label{prop:GuptaBleuler}
Suppose that $A,B$ are complex solutions to the source-free Maxwell equation\footnote{For making this Proposition rigorous, we would also have to assume that $A$ and $B$ have compact support on every Cauchy surface or at least decay sufficiently fast at infinity.}
\be\label{freeMaxwell}
\partial^\mu (\partial_\mu A_\nu - \partial_\nu A_\mu)=0
\ee
The scalar product \eqref{scp1yc} is not positive definite, but 
(i)~it is sesqui-linear and Hermitian,
\be
\scp{B}{A}_\Sigma = \scp{A}{B}_\Sigma^*\,;
\ee
(ii)~it is independent of $\Sigma$, i.e., for any two Cauchy surfaces $\Sigma,\Sigma'$,
\be
\scp{A}{B}_{\Sigma'} = \scp{A}{B}_\Sigma\,;
\ee
and (iii)~it is invariant under gauge transformations $A_\mu \to A_\mu - \tfrac{1}{e_x}\partial_\mu \theta$, $B_\mu \to B_\mu-\tfrac{1}{e_x}\partial_\mu \zeta$.
\end{prop}

\begin{proof}
(i) is easily verified. Concerning (ii), it suffices to consider $\Sigma'$ in the future of $\Sigma$ (or else consider a third Cauchy surface $\Sigma''$ in the future of both). We begin by computing the 4-divergence of the square bracket in \eqref{scp1yc}; the underlined terms cancel:
\begin{subequations}\label{4divGuptaBleuler}
\begin{align}
\partial^\nu[\cdots]_\nu
&=\underline{\partial^\nu A_\mu^* \partial_\nu B^\mu} + A_\mu^* \partial^\nu\partial_\nu B^\mu-\underline{\underline{\partial^\nu A_\nu^* \partial_\mu B^\mu}} -A_\nu^* \partial^\nu\partial_\mu B^\mu \nonumber\\
&- \partial^\nu\partial_\nu A_\mu^* B^\mu-\underline{\partial_\nu A_\mu^* \partial^\nu B^\mu}+\partial^\nu \partial^\mu A_\mu^* B_\nu+\underline{\underline{\partial^\mu A_\mu^* \partial^\nu B_\nu}}\\
&=A_\mu^* \partial^\nu\partial_\nu B^\mu-A_\nu^* \partial^\nu\partial_\mu B^\mu
- \partial^\nu\partial_\nu A_\mu^* B^\mu+\partial^\nu \partial^\mu A_\mu^* B_\nu\\
&=A_\mu^* \partial^\nu\partial_\nu B^\mu-A_\mu^* \partial^\mu\partial_\nu B^\nu
- \partial^\nu\partial_\nu A_\mu^* B^\mu+\partial_\mu \partial^\nu A_\nu^* B^\mu\\
&=A_\mu^* \underbrace{\partial_\nu(\partial^\nu B^\mu- \partial^\mu B^\nu)}_{=0 \text{ by \eqref{freeMaxwell}}}
- \underbrace{\partial^\nu(\partial_\nu A_\mu^* -\partial_\mu A_\nu^*)}_{=0 \text{ by \eqref{freeMaxwell}}} B^\mu=0.
\end{align}
\end{subequations}
According to a version of the Ostrogradski-Gauss integral theorem, for a vector field $j^\nu(x)$ and a 4d region $V$ with boundary $\partial V$,
\be\label{Gauss}
\int_{\partial V} V(d^3x) \, n^\nu(x) \, j_\nu(x) =
\int_V d^4x \, \partial^\nu j_\nu(x) 
\ee
with $n^\nu$ the outward unit normal vector. For $\Sigma'$ in the future of $\Sigma$, it follows that
\be\label{Gauss2}
\int_{\Sigma'}V(d^3x) \, n^\nu(x) \, j_\nu(x)
-\int_{\Sigma}V(d^3x) \, n^\nu(x) \, j_\nu(x)
=\int_L d^4x \, \partial^\nu j_\nu(x)
\ee
with $n^\nu$ the future unit normal vector on either $\Sigma'$ or $\Sigma$ (depending on which surface is being considered) and $L$ the layer between $\Sigma$ and $\Sigma'$. Then, taking for $j_\nu$ the square bracket in \eqref{scp1yc}, \eqref{4divGuptaBleuler} yields that $\scp{A}{B}_\Sigma$ is independent of $\Sigma$.  

(iii) Gauge invariance: It suffices to show that $\scp{A}{\partial \theta}_\Sigma=0$ for any scalar function $\theta$ (that has compact support on every Cauchy surface or decays sufficiently fast at infinity together with its gradient). By $\Sigma$-independence,
\begin{subequations}
\begin{align}
	i\scp{A}{\partial \theta}_\Sigma 
	&=  \int_\Sigma V(d^3x)\, n^\nu(x) \,\Bigl[ A_\mu^* \partial_\nu \partial^\mu \theta -A_\nu^* \partial_\mu \partial^\mu \theta - (\partial_\nu A_\mu^*) \partial^\mu \theta + (\partial^\mu A_\mu^*)\partial_\nu \theta \Bigr]\\
	&=  \int_{t=0} \!\! d^3\vx ~\Bigl[ A_\mu^* \partial_0 \partial^\mu \theta -A_0^* \partial_\mu \partial^\mu \theta - (\partial_0 A_\mu^*) \partial^\mu \theta + (\partial^\mu A_\mu^*)\partial_0 \theta \Bigr]\\
	&=  \int_{t=0} \!\! d^3\vx ~\Bigl[ \underline{A_i^* \partial_0 \partial^i \theta} -A_0^* \partial_i \partial^i \theta - (\partial_0 A_i^*) \partial^i \theta + \underline{(\partial^i A_i^*)\partial_0 \theta} \Bigr]\\
	&=  \int_{t=0} \!\! d^3\vx ~ (\partial_i A_0^* - \partial_0 A_i^*) \partial^i \theta \\
	&=  \int_{t=0} \!\! d^3\vx ~ F_{i0}^*\, \partial^i \theta \\
	&=  \int_{t=0} \!\! d^3\vx ~ (-\partial^i F_{i0}^*)\, \theta \\
	&=  \int_{t=0} \!\! d^3\vx ~ \Bigl(\underbrace{\partial^\mu F_{\mu 0}^*}_{\text{=0 by \eqref{freeMaxwell}}}-\partial^0 \underbrace{F_{00}^*}_{=0} \Bigr)\, \theta \\
	&=0,
\end{align}
\end{subequations}
where Latin indices $i$ run from 1 to 3 while Greek indices $\mu$ run from 0 to 3.
\end{proof}

\subsection{Special Cases}

We describe some special cases of \eqref{scp1yc}. If we assume the Lorenz gauge condition
\be\label{Lorenz1}
\partial^\mu A_\mu=0
\ee
of both $A$ and $B$, then \eqref{scp1yc} simplifies to
\be
\scp{A}{B}_{1y\Sigma} 
= -i \int_\Sigma V(d^3x) \: n^\nu(x) ~\Bigl[ A_\mu^* \partial^{~}_\nu B^\mu- (\partial^{~}_\nu A_\mu^*) B^\mu \Bigr]\,. \label{scp1yd}
\ee
If we set $\Sigma=\{t=0\}$ in a given Lorentz frame, this simplifies further to
\be
\scp{A}{B}_{1y0} 
= -i \int_{t=0} \!\! d^3\vx \: \Bigl[ A_\mu^* \partial^{~}_0 B^\mu- (\partial^{~}_0 A_\mu^*) B^\mu \Bigr]\,. \label{scp1ye}
\ee
Together with the Lorenz condition \eqref{Lorenz1}, the source-free Maxwell equation \eqref{freeMaxwell} implies that $\square A_\mu=0$, so the 4d Fourier transform $\hat A_\mu(k)$ of $A_\mu(x)$ is concentrated on the light cone $k^\nu k_\nu=0$. If we allow only positive energy $k^0$, then the latter is determined according to 
\be\label{positive1}
k^0 = |\vk|\,,
\ee
and we can express $\hat A_\mu(k)=\hat A_\mu(|\vk|,\vk)$ as a function $\hat A_\mu(\vk)$ of $\vk$. Then, $A_\mu(t=0,\vx)$ is the 3d Fourier transform of $\hat A_\mu(\vk)$, $A_\mu(t=T,\vx)$ that of $e^{-iT|\vk|} \hat A_\mu(\vk)$, and $\partial_0 A_\mu(t=0,\vx)$ that of $-i|\vk| \hat A_\mu(\vk)$. Since 3d Fourier transformation is unitary relative to the non-relativistic scalar product $\int d^3\vx \, \psi^*(\vx)\, \phi(\vx)$, we obtain from \eqref{scp1ye} that\footnote{Note that $\hat A^*$ is the conjugate of the Fourier transform, not the Fourier transform of the conjugate.}
\begin{subequations}
\begin{align}
\scp{A}{B}_{1y0} 
&= -i \int_{\RRR^3} d^3\vk \: \Bigl[ \hat A_\mu^*(\vk)\, (-i)|\vk| \hat B^\mu(\vk)
- i|\vk| A_\mu^*(\vk)\, \hat B^\mu(\vk) \Bigr]\\
&= -2\int_{\RRR^3} d^3\vk \: \hat A_\mu^*(\vk)\, |\vk| \, \hat B^\mu(\vk)
\,. \label{scp1yf}
\end{align}
\end{subequations}
This expression is positive semidefinite; indeed, the Lorenz condition \eqref{Lorenz1} requires that $\hat A_\mu(k)\, k^\mu=0$, so $\hat A_\mu$ is either spacelike or a multiple of $k_\mu$, with the consequence that either $\hat A_\mu \hat A^\mu <0$ or $\hat A_\mu \hat A^\mu =0$; thus, since $|\vk|\geq 0$, $\scp{A}{A}_{1y0}\geq 0$. 

Moreover, this reasoning shows that $\scp{A}{A}_{1y0}=0$ only if $\hat A_\mu \hat A^\mu =0$ for almost all $\vk$, which implies that $\hat A_\mu$ is a multiple of $k_\mu$, and thus a gradient of some scalar function $\theta$. Therefore, on the quotient space of $A_\mu$'s modulo pure gauge fields, $\scp{\cdot}{\cdot}_{1y0}$ is positive definite and thus defines an inner product in the standard sense. This is the inner product used in connection with Gupta-Bleuler quantization \cite{GB}. 

\begin{prop}
Assuming the Lorenz condition \eqref{Lorenz1} and positive energy \eqref{positive1}, \eqref{scp1yf} can be expressed in terms of $F_{\mu\nu}$ instead of $A_\mu$ as 
\be
\scp{A}{B}_{1y0} 
= \int_{\RRR^3} \frac{d^3\vk}{|\vk|} \sum_{\mu,\nu=0}^3 \hat F_{\mu\nu}^*(\vk) \, \hat G_{\mu\nu}(\vk)
\,, \label{scp1yg}
\ee
where $G_{\mu\nu}$ is obtained from $B_\mu$ in the same way as $F_{\mu\nu}$ from $A_\mu$ and the hat means again 3d Fourier transformation. 
\end{prop}

\begin{proof}
Consider first $B=A$. If we write $k^\mu=(|\vk|,\vk)$ and $A_\mu=(A^0,-\vA)$, then the Lorenz condition \eqref{Lorenz1} translates to $\hat A^0=\vk\cdot\hat\vA/|\vk|$, so $\hat A_\mu^* \hat A^\mu = -|P_{\vk}^\perp \hat\vA|^2$ with $P_{\vk}^\perp$ the $3\times 3$ projection matrix to the subspace orthogonal to $\vk$. We find that $\hat F_{0i}= i|\vk|\hat A^i - i k^i\vk\cdot\hat\vA/|\vk|$ and $\hat F_{ij} = -i k^i \hat A^j + i k^j \hat A^i$, so
\be
\sum_{i=1}^3 |\hat F_{0i}|^2=|\vk|^2 |P_{\vk}^\perp \hat\vA|^2=\tfrac{1}{2}\sum_{i,j=1}^3|\hat F_{ij}|^2
\ee
and thus
\be
\scp{A}{A}_{1y0} 
= \int_{\RRR^3} \frac{d^3\vk}{|\vk|} \sum_{\mu,\nu=0}^3 |\hat F_{\mu\nu}(\vk)|^2
\,. \label{scp1yh}
\ee
Now we turn to the case $B\neq A$. For any sesqui-linear form $\scp{\cdot}{\cdot}$, its values can be recovered from expressions of the form $\scp{\varphi}{\varphi}$ (i.e., product of a function with itself) by means of the \emph{polarization identity}
\be\label{polarization}
\scp{\psi}{\phi} = \tfrac{1}{4} \Bigl(\scp{\psi+\phi}{\psi+\phi} - \scp{\psi-\phi}{\psi-\phi} -i \scp{\psi+i\phi}{\psi+i\phi} +i\scp{\psi-i\phi}{\psi-i\phi}  \Bigr).
\ee
Since the right-hand side of \eqref{scp1yg} is a sesqui-linear expression that agrees with the left-hand side for $B=A$, it must also agree for $B\neq A$, which completes the proof of \eqref{scp1yg}.
\end{proof}

The expressions \eqref{scp1yg} and \eqref{scp1yh} correspond to 8 times the inner product characterized by Landau and Peierls in Eq.s (8) and (10) of \cite{LP30}. As they say, the choice is motivated by the thought that $\sum_{\mu\nu} |\hat F_{\mu\nu}|^2$ is 8 times the energy according to classical electrodynamics, while the energy of a single photon is $k^0=|\vk|$; for further discussion see Eq.~(7.3.57) in \cite{Tum22}. This inner product is also discussed in \cite[Sec.~5.1]{BB}.

\section{Proof of Proposition~\ref{prop:unitary}: Unitarity}
\label{app:unitary}

\begin{proof}
Since any sesqui-linear form $\scp{\cdot}{\cdot}$ can be recovered from the quadratic expressions $\scp{\varphi}{\varphi}$ by the polarization identity \eqref{polarization}, it suffices to prove that $\scp{\Psi}{\Psi}_\Sigma=\scp{\Psi}{\Psi}_{\Sigma'}$, and for that, it suffices to consider $\Sigma'$ in the future of $\Sigma$ (or else consider a third surface $\Sigma''$ in the future of both). Let $L$ denote the layer in space-time between $\Sigma$ and $\Sigma'$, and $\sJ=\sJ_{\nu_1...\nu_m\nu'_1...\nu'_n}(x_1...x_m,y_1...y_n)$ the sum over $s_1...s_m,s'_1...s'_m$ in \eqref{scpdef} with $\Phi$ replaced by $\Psi$. 

We now use the Ostrogradski-Gauss integral theorem in the form \eqref{Gauss2} for $j=\sJ$; we can extend it as follows for an integrand of the form $j_{\nu_1\nu_2}(x_1,x_2)$:
\begin{subequations}
\begin{align}
&\int_{\Sigma'}V(d^3x_1) \int_{\Sigma'} V(d^3x_2) \,n^{\nu_1}(x_1) \, n^{\nu_2}(x_2) \, j_{\nu_1\nu_2}(x_1,x_2) \nonumber\\
&~~~~~~~~~~ = \int_L d^4x_1 \int_{\Sigma'} V(d^3x_2) \, n^{\nu_2}(x_2) \, \partial^{\nu_1}_{x_1} j_{\nu_1\nu_2}(x_1,x_2) \nonumber \\
&~~~~~~~~~~~~ + \int_\Sigma V(d^3x_1) \int_{\Sigma'} V(d^3x_2) \,n^{\nu_1}(x_1)\, n^{\nu_2}(x_2) \, j_{\nu_1\nu_2}(x_1,x_2)  \\
&~~~~~~~~~~ = \int_L d^4x_1 \int_{\Sigma'} V(d^3x_2) \, n^{\nu_2}(x_2) \, \partial^{\nu_1}_{x_1} j_{\nu_1\nu_2}(x_1,x_2) \nonumber \\
&~~~~~~~~~~~~ + \int_\Sigma V(d^3x_1) \int_{L} d^4x_2 \,n^{\nu_1}(x_1)\, \partial^{\nu_2}_{x_2} j_{\nu_1\nu_2}(x_1,x_2) \nonumber \\
&~~~~~~~~~~~~ + \int_\Sigma V(d^3x_1) \int_{\Sigma} V(d^3x_2) \,n^{\nu_1}(x_1)\, n^{\nu_2}(x_2) \, j_{\nu_1\nu_2}(x_1,x_2) \,.
\end{align}
\end{subequations}
We abbreviate this equation as
\be
\int_{\Sigma'} \int_{\Sigma'} = \int_L \int_{\Sigma'} + \int_\Sigma \int_L + \int_\Sigma \int_\Sigma\,.
\ee
Likewise, we can resolve the $\Sigma'$ integral on the right-hand side further to obtain
\be
\int_{\Sigma'} \int_{\Sigma'} = \int_L \int_L + \int_L \int_{\Sigma} + \int_\Sigma \int_L + \int_\Sigma \int_\Sigma\,.
\ee
Now for our purposes, it suffices to consider $\Sigma'$ just infinitesimally far away from $\Sigma$; say, $d\tau$ is an infinitesimal quantity, and $x^\mu+ n^\mu(x) \, \theta(x) \, d\tau \in \Sigma'$ for every $x\in\Sigma$; that is, $\theta(x)\, d\tau$ is the thickness of $L$ at $x$. Then,
\be\label{intLintSigma}
\int_L d^4x \, f(x) = \biggl(\int_\Sigma V(d^3x) \, \theta(x) \, f(x) \biggr) d\tau + o(d\tau)\,.
\ee
Since every integral over $L$ is small like $d\tau$, the $\int_L \int_L$ term is small like $d\tau^2$; that is,
\be
\int_{\Sigma'} \int_{\Sigma'}- \int_\Sigma \int_\Sigma = \int_L \int_{\Sigma} + \int_\Sigma \int_L + \, O(d\tau^2) \,,
\ee
a kind of Leibniz rule. Likewise, for any number $n$ of 4d variables, 
\be\label{Leibnizn}
\int_{\Sigma^{\prime n}}-\int_{\Sigma^n} = \sum_{j=1}^n \int_{\Sigma^{j-1}} \int_L ~~\int_{\Sigma^{n-j}} + \, O(d\tau^2)\,.
\ee

As a preparation for the steps to follow, we compute the divergence of $\sJ$ in each 4d variable, $\partial_{x_i}^{\nu_i}\sJ_{\nu_i}$ and $\partial_{y_\ell}^{\nu'_\ell}\sJ_{\nu'_\ell}$. Let
\be
\sD=i\gamma^\mu\partial_\mu-m_x
\ee
denote the free Dirac operator, and recall that for a 1-particle Dirac wave function $\psi_s(x)$, since $\gamma^0$ and $\gamma^0\gamma_\nu$ are self-adjoint,
\begin{subequations}
\begin{align}
\partial^\nu(\overline{\psi}\gamma_\nu \psi)
&=i\overline{i\gamma^\nu\partial_\nu\psi} \psi 
-i \overline{\psi}i\gamma_\nu \partial^\nu \psi 
\\
&=i \overline{\sD \psi} \psi -i \overline{\psi} \sD \psi\\
&=2\,\Im(\overline{\psi}\sD\psi)\,.
\end{align}
\end{subequations}
For similar reasons, using \eqref{LPx} and that
\be\label{Dconjugate}
(A^*_\mu D^{\mu\mu'}_\nu B_{\mu'})^* = B^*_{\mu} D^{\mu\mu'}_\nu A_{\mu'} 
\ee
with $D$ as in \eqref{Dmumunudef}, we obtain for our case that
\begin{align}
\partial_{x_i}^{\nu_i}\sJ_{\nu_i}
&=2\,\Im\Bigl[\Psi^\dagger \gamma^0_i e_x\sqrt{n+1}\gamma_i^\rho \Biggl[ \prod_{j\neq i} \gamma_j^0\gamma_{j\nu_j} \Biggr] \Biggl[ \prod_{k=1}^n D_k \Biggr] \Psi^{(m,n+1)}_{\mu_{n+1}=\rho}(y_{n+1}=x_i)\Bigr] \label{xdivJ}
\end{align}
with $D_k:=D^{\mu_k\mu'_k}_{y_k,\nu'_k}$.

Concerning the $y$-divergence of $\sJ$, note first that
\begin{subequations}
\begin{align}
\partial^\nu(A^*_\mu D^{\mu\mu'}_\nu A_{\mu'})
&= -i\partial^\nu\Bigl[ A_\mu^* \partial^{~}_\nu A^\mu-A_\nu^* \partial^{~}_\mu A^\mu - (\partial^{~}_\nu A_\mu^*) A^\mu+(\partial^\mu A_\mu^*)A^{~}_\nu \Bigr]\\
&= 2\partial^\nu \Im\Bigl[ A_\mu^* \partial^{~}_\nu A^\mu-A_\nu^* \partial^{~}_\mu A^\mu \Bigr]\\
&= 2\, \Im\Bigl[ \underbrace{\partial^\nu A_\mu^* \partial^{~}_\nu A^\mu}_{\text{real}}
+A_\mu^* \partial^\nu \partial^{~}_\nu A^\mu
-\underbrace{\partial^\nu A_\nu^* \partial^{~}_\mu A^\mu}_{\text{real}} -A_\nu^* \partial^\nu\partial^{~}_\mu A^\mu \Bigr]\\
&= 2\, \Im\Bigl[ 
A^{\nu*} \partial^\mu (\partial^{~}_\mu A_\nu
-\partial^{~}_\nu A_\mu) \Bigr]\,.
\end{align}
\end{subequations}
For similar reasons and using \eqref{Dconjugate} and \eqref{LPy}, here
\begin{align}
\partial_{y_\ell}^{\nu'_\ell}\sJ_{\nu'_\ell}
&= 2\, \Im \Bigl[(\Psi^{\mu_\ell=\rho})^\dagger 
\Biggl[ \prod_j \gamma_j^0\gamma_{j\nu_j} \Biggr] \Biggl[ \prod_{k\neq \ell} D_k \Biggr] 
\frac{e_x}{\sqrt{n}}\sum_{j=1}^m \delta^3_{\mu}(y_\ell-x_j) \gamma_j^\mu \gamma_{j\rho} \Psi^{(m,n-1)}_{\widehat{\mu_\ell}}(\widehat{y_\ell})
 \Bigr]\,. \label{ydivJ}
\end{align}
Inserting this and \eqref{xdivJ} in \eqref{Leibnizn}, we obtain from \eqref{scpdef} that
\begin{subequations}
\begin{align}
\scp{\Psi}{\Psi}_{\Sigma'}-\scp{\Psi}{\Psi}_\Sigma
&= \sum_{m=1}^\infty \sum_{n=0}^\infty \sum_{i=1}^m \int_L d^4x_i\int_{\Sigma^{m+n-1}} \hspace{-9mm} V(d^3x_1) \cdots \widehat{V(d^3x_i)} \cdots V(d^3y_n)\:\times \nonumber\\
&~~~~~ \times n(x_1)\cdots \widehat{n(x_i)} \cdots n(y_n) \partial_{x_i}^{\nu_i}\sJ_{\nu_i} \nonumber\\
& + \sum_{m=0}^\infty \sum_{n=1}^\infty \sum_{\ell=1}^n \int_L d^4y_\ell  \int_{\Sigma^{m+n-1}}\hspace{-9mm} V(d^3x_1) \cdots \widehat{V(d^3y_\ell)} \cdots V(d^3y_n)\:\times \nonumber\\
&~~~~~ \times n(x_1)\cdots \widehat{n(y_\ell)} \cdots n(y_n) \partial_{y_\ell}^{\nu'_\ell}\sJ_{\nu'_\ell} + O(d\tau^2)\\
&= \sum_{m=1}^\infty \sum_{n=0}^\infty \sum_{i=1}^m \int_L d^4x_i\int_{\Sigma^{m+n-1}} \hspace{-9mm} V(d^3x_1) \cdots \widehat{V(d^3x_i)} \cdots V(d^3y_n)\:\times \nonumber\\
&~~~~~ \times n(x_1)\cdots \widehat{n(x_i)} \cdots n(y_n) \, 2\,\Im\Bigl[\Psi^\dagger \gamma^0_i e\sqrt{n+1}\gamma_i^\rho \:\times \nonumber\\
&~~~~~ \times \Biggl[ \prod_{j\neq i} \gamma_j^0\gamma_{j\nu_j} \Biggr] \Biggl[ \prod_{k=1}^n D_k \Biggr] \Psi^{(m,n+1)}_{\mu_{n+1}=\rho}(y_{n+1}=x_i)\Bigr]  \nonumber\\
& + \sum_{m=0}^\infty \sum_{n=0}^\infty (n+1) \int_L d^4y_{n+1}  \int_{\Sigma^{m+n}}\hspace{-6mm} V(d^3x_1) \cdots V(d^3y_n)\:\times \nonumber\\
&~~~~~ \times n(x_1)\cdots n(y_n) \, 2\, \Im \Bigl[(\Psi^{(m,n+1)\mu_{n+1}=\rho})^\dagger \Biggl[ \prod_{j=1}^m \gamma_j^0\gamma_{j\nu_j} \Biggr] \:\times \nonumber\\
&~~~~~ \times\Biggl[ \prod_{k=1}^n D_k \Biggr] 
\frac{e_x}{\sqrt{n+1}}\sum_{j=1}^m \delta^3_{\mu}(y_{n+1}-x_j) \gamma_j^\mu \gamma_{j\rho} \Psi^{(m,n)}
 \Bigr] + O(d\tau^2)
\end{align}
[where we have renamed $n\to n+1$ and $\ell\to n+1$ in the last sum; the term $m=0$ vanishes because it has empty $\sum_j$; now use \eqref{intLintSigma}, take the conjugate inside the last $\Im[\ldots]$, and use \eqref{Dconjugate} and that $\gamma^0\gamma_\nu$ is self-adjoint:]
\begin{align}
&= d\tau\sum_{m=1}^\infty \sum_{n=0}^\infty \sum_{i=1}^m \int_\Sigma V(d^3x_i) \, \theta(x_i)\int_{\Sigma^{m+n-1}} \hspace{-9mm} V(d^3x_1) \cdots \widehat{V(d^3x_i)} \cdots V(d^3y_n)\:\times \nonumber\\
&~~~~~ \times n(x_1)\cdots \widehat{n(x_i)} \cdots n(y_n) 2e_x\sqrt{n+1}\,\Im\Bigl[\Psi^\dagger \gamma^0_i \gamma_i^\rho \:\times \nonumber\\
&~~~~~ \times \Biggl[ \prod_{j\neq i} \gamma_j^0\gamma_{j\nu_j} \Biggr] \Biggl[\prod_{k=1}^n D_k \Biggr] \Psi^{(m,n+1)}_{\mu_{n+1}=\rho}(y_{n+1}=x_i)\Bigr]  \nonumber\\
& - d\tau\sum_{m=1}^\infty \sum_{n=0}^\infty  \int_\Sigma V(d^3y_{n+1}) \, \theta(y_{n+1})  \int_{\Sigma^{m+n}}\hspace{-6mm} V(d^3x_1) \cdots V(d^3y_n)\:\times \nonumber\\
&~~~~~ \times n(x_1)\cdots n(y_n) 2e_x \sqrt{n+1}\, \Im \Bigl[  \sum_{i=1}^m \delta^3_{\mu}(y_{n+1}-x_i) \Bigl( \gamma_i^\mu \gamma_i^\rho \Psi^{(m,n)}\Bigr)^\dagger  \:\times \nonumber\\
&~~~~~ \times\Biggl[ \prod_{k=1}^n D_k \Biggr] \Biggl[
\prod_j \gamma_j^0\gamma_{j\nu_j} \Biggr] \Psi^{(m,n+1)}_{\mu_{n+1}=\rho}
 \Bigr] + o(d\tau)
\end{align}
[simplifying]
\begin{align}
&= d\tau\sum_{m=1}^\infty \sum_{n=0}^\infty \sum_{i=1}^m  \int_{\Sigma^{m+n}} \hspace{-6mm} V(d^3x_1)  \cdots V(d^3y_n) \, \theta(x_i)\:\times \nonumber\\
&~~~~~ \times n(x_1)\cdots \widehat{n(x_i)} \cdots n(y_n) 2e_x\sqrt{n+1}\,\Im\Bigl[\Psi^\dagger \gamma^0_i \gamma_i^\rho \:\times \nonumber\\
&~~~~~ \times \Biggl[ \prod_{j\neq i} \gamma_j^0\gamma_{j\nu_j} \Biggr] \Biggl[ \prod_{k=1}^n D_k\Biggr] \Psi^{(m,n+1)}_{\mu_{n+1}=\rho}(y_{n+1}=x_i)\Bigr]  \nonumber\\
& - d\tau\sum_{m=1}^\infty \sum_{n=0}^\infty \sum_{i=1}^m \int_{\Sigma^{m+n}}\hspace{-6mm} V(d^3x_1) \cdots V(d^3y_n)\, \theta(x_i)\, 2e_x \sqrt{n+1} \:\times \nonumber\\
&~~~~~ \times n(x_1)\cdots \widehat{n(x_i)}\cdots n(y_n) \, \Im \Bigl[ \Bigl( \gamma_i^0 \underbrace{n^{\nu_i}(x_i) \gamma_{i\nu_i}  n_{\mu}(x_i) \gamma_i^\mu}_{=I} \gamma_i^\rho \Psi^{(m,n)}\Bigr)^\dagger  \:\times \nonumber\\
&~~~~~ \times \Biggl[\prod_{j\neq i} \gamma_j^0\gamma_{j\nu_j} \Biggr]
\Biggl[ \prod_{k=1}^n D_k \Biggr]  \Psi^{(m,n+1)}_{\mu_{n+1}=\rho}(y_{n+1}=x_i)
 \Bigr] + o(d\tau)\\
&= d\tau\sum_{m=1}^\infty \sum_{n=0}^\infty \sum_{i=1}^m  \int_{\Sigma^{m+n}} \hspace{-6mm} V(d^3x_1)  \cdots V(d^3y_n) \, \theta(x_i)\:\times \nonumber\\
&~~~~~ \times n(x_1)\cdots \widehat{n(x_i)} \cdots n(y_n) 2e_x\sqrt{n+1}\,\Im\Bigl[\Psi^\dagger \gamma^0_i \gamma_i^\rho \:\times \nonumber\\
&~~~~~ \times \Biggl[ \prod_{j\neq i} \gamma_j^0\gamma_{j\nu_j} \Biggr] \Biggl[ \prod_{k=1}^n D_k\Biggr] \Psi^{(m,n+1)}_{\mu_{n+1}=\rho}(y_{n+1}=x_i)\Bigr]  \nonumber\\
& - d\tau\sum_{m=1}^\infty \sum_{n=0}^\infty \sum_{i=1}^m \int_{\Sigma^{m+n}}\hspace{-6mm} V(d^3x_1) \cdots V(d^3y_n)\, \theta(x_i)\, 2e_x \sqrt{n+1} \:\times \nonumber\\
&~~~~~ \times n(x_1)\cdots \widehat{n(x_i)}\cdots n(y_n) \, \Im \Bigl[ \Psi^\dagger \gamma_i^0  \gamma_i^\rho  \:\times \nonumber\\
&~~~~~ \times \Biggl[\prod_{j\neq i} \gamma_j^0\gamma_{j\nu_j} \Biggr]
\Biggl[ \prod_{k=1}^n D_k \Biggr]  \Psi^{(m,n+1)}_{\mu_{n+1}=\rho}(y_{n+1}=x_i)
 \Bigr] + o(d\tau)\\
 &=0\, d\tau + o(d\tau)\,,
\end{align}
\end{subequations}
which is what we needed to show.\footnote{For making this proof mathematically rigorous, one would have to investigate in particular under which conditions it is true of the series $\sum_{n=1}^\infty$, in which every term is $o(d\tau)$, that also the sum is $o(d\tau)$, as taken for granted here. We expect that this is true once \eqref{LPg} for some gauge condition $\sA$ is assumed, but it cannot be true if only \eqref{LPx} and \eqref{LPy} are assumed, as the example of Remark~\ref{rem:indispensable} in Section~\ref{sec:rem2} shows. We thank Lukas Nullmeier for drawing our attention to this point.}
\end{proof}

\section{Proof of Proposition~\ref{prop:constraint}: Landau and Peierls' Constraint Condition}
\label{app:constraint}

\begin{proof}
Consider the Maxwell equations \eqref{MaxwellEB} and Fourier transform them in $\vx$ (but not in $t=x^0$) to obtain
\begin{subequations}
\begin{align}
\partial_t \hat\vE(t,\vk)&=i\vk\times \hat\vB(t,\vk) -\hat\vJ(t,\vk) \label{hatEt}\\
i\vk\cdot \hat\vE(t,\vk)&= \hat{J}_0(t,\vk)\\
\partial_t \hat\vB(t,\vk) &= -i\vk\times \hat\vE(t,\vk) \label{hatBt}\\
\vk\cdot\hat\vB(t,\vk)&=0\,.
\end{align}
\end{subequations}
Split the vector fields ($\hat\vE,\hat\vB,\hat\vJ$) into parts parallel ($E_{||}$ etc.) and perpendicular ($\vE_\perp$ etc.) to $\vk$, so $\hat\vE=E_{||} \vk/k + \vE_\perp$ etc., where $k=|\vk|$; then
\begin{subequations}
\begin{align}
\partial_t \vE_\perp&=i\vk\times \vB_\perp -\vJ_\perp \label{Eperpt}\\
E_{||}&= -i\hat{J}_0/k\\
\partial_t \vB_\perp &= -i\vk\times \vE_\perp \label{Bperpt}\\
B_{||}&=0\,.
\end{align}
\end{subequations}
In this version, the equations for different $\vk$ values decouple, so one obtains an ODE for every $\vk$. It not hard to derive that for $\vJ_\perp=0$, the general solution for the evolving parts $\vE_\perp,\vB_\perp$ is
\be
\begin{pmatrix} \vE_\perp(t,\vk)\\ \vB_\perp(t,\vk)\end{pmatrix}
= e^{-ikt} \begin{pmatrix} \vE^-(\vk) \\ \vk\times \vE^-(\vk)/k \end{pmatrix}
+ e^{ikt} \begin{pmatrix} \vE^+(\vk) \\ -\vk\times\vE^+(\vk)/k \end{pmatrix}
\ee
with arbitrary $t$-independent $\vk$-perpendicular vector fields $\vE^{\pm}$ (note that $-(\vk\times)^2/k^2$ is the projection to the subspace perpendicular to $\vk$). The part with $e^{-ikt}$ has negative frequency and thus positive energy, and it satisfies
\be\label{constraintk}
\vB_\perp=\vk\times\vE_\perp/k
\ee
or \eqref{constraint}. In general, however, \eqref{constraintk} is not preserved, as can be seen from taking the time derivative on both sides and using \eqref{Eperpt} and \eqref{Bperpt}, which implies $-i\vk\times\vJ_\perp/k=0$. 
\end{proof}

\section{Details about Conjecture 1}
\label{app:conj}

In order to define what it means for a Hamiltonian $H$ to obey \emph{propagation locality}, we demand that if a wave function $\psi$ is concentrated in a region $A\subseteq \RRR^3$, then $e^{-iHt}\psi$ is concentrated for every $t\in\RRR$ in the ``grown set''
\be
\Gr(A,|t|) = \bigcup_{\vx\in A} \overline{B_{|t|}}(\vx)
\ee
with $\overline{B_r}(\vx)=\{\vy\in\RRR^3:|\vx-\vy|\leq r\}$ the closed ball of radius $r$ and center $\vx$. Here, $\psi$ is ``concentrated in $A$'' if it vanishes at any configuration containing at least one particle outside of $A$. We define further that such $H$ obeys \emph{interaction locality} if for every wave function $\psi$ concentrated in $A=A_1\cup A_2\subset \RRR^3$ with $A_1$ and $A_2$ separated by distance $d>0$ and for every $t\in\RRR$ with $|t|<d/2$, $e^{-iHt}\psi = [U_1(t) \otimes U_2(t)]\psi$ for suitable operators $U_1(t)$ acting only on $\Gr(A_1,|t|)$ and $U_2(t)$ acting only on $\Gr(A_2,|t|)$. The Hamiltonian $H$ in the Hilbert space $\Hilbert$ will count as \emph{Poincar\'e invariant} if it is part of a unitary representation of the proper Poincar\'e group on $\Hilbert$, i.e., if $H$ is the generator of the representation of the subgroup formed by the translations in $x^0$ direction. In order to define what a general, arbitrary \emph{multi-time evolution} is, it is convenient to consider, instead of PDEs such as \eqref{LPx} and \eqref{LPy} in terms of the space-time coordinates of spacelike configurations, a \emph{hypersurface evolution} \cite{LT:2017}; that is a family of Hilbert spaces $\Hilbert_\Sigma$ associated with every Cauchy surface $\Sigma$, equipped with a PVM $P_\Sigma$ on the configuration space $\Gamma(\Sigma)$ (i.e., the set of all finite subsets of $\Sigma$) acting on $\Hilbert_\Sigma$ and unitary isomorphisms $U_\Sigma^{\Sigma'}:\Hilbert_\Sigma\to \Hilbert_{\Sigma'}$ for any two Cauchy surfaces. Definitions of \emph{interaction locality} and \emph{propagation locality} for a hypersurface evolution have been formulated in \cite{LT:2017}. It is \emph{Poincar\'e invariant} if equipped, for every element $g$ of the proper Poincar\'e group and every Cauchy surface $\Sigma$, with a unitary isomorphism $S_{g,\Sigma}: \Hilbert_\Sigma \to \Hilbert_{g\Sigma}$ (representing the action of $g$ on wave functions) that fits together with $S_{h,g\Sigma}$, $U_{\Sigma}^{\Sigma'}$, $U_{g\Sigma}^{g\Sigma'}$, $P_\Sigma$, and $P_{g\Sigma}$ in the appropriate way; for more discussion see \cite{LT:2021b}.

\bigskip

\noindent{\it Acknowledgments.} We thank Lukas Nullmeier for helpful discussions.

\section*{Declarations}

\noindent{\it Funding.} This research received no funding.

\noindent{\it Conflict of interests.} The authors declare no conflict of interest.

\noindent{\it Availability of data and material.} Not applicable.

\noindent{\it Code availability.} Not applicable.

\end{document}